\pdfoutput=1
\documentclass{article}
\usepackage{arxiv}
\usepackage[utf8]{inputenc}
\usepackage[english]{babel}
\usepackage[T1]{fontenc}  
\usepackage{extarrows}
\usepackage{blindtext}
\usepackage{hyperref}       
\usepackage{url}           
\usepackage{booktabs}     
\usepackage{amsfonts} 
\usepackage{mathrsfs}
\usepackage{amssymb}
\usepackage{amsthm}
\usepackage{amsmath}
\usepackage{amsfonts} 
\usepackage{nicefrac}       
\usepackage{microtype}      
\usepackage{graphicx}

\title{Wall Crossing Structures and Application to SU(3) Seiberg-Witten Integrable system}
\author{
  Qiang Wang\thanks{wangqiang@math.ksu.edu\, Department of Mathematics, Kansas State University, Manhattan, KS, 66502} 
}

\begin{document}
\maketitle

\newtheorem{theorem}{Theorem}[section]
\newtheorem{corollary}{Corollary}[theorem]
\newtheorem{lemma}[theorem]{Lemma}
\newtheorem{remark}{Remark}[section]
\newtheorem{proposition}[theorem]{Proposition}
\newtheorem{definition}{Definition}[section]

\begin{abstract}
    We apply the \textit{wall crossing structure} formalism of Kontsevich and Soibelman to Seiberg-Witten integrable systems associated to pure $SU(3)$. This gives an algorithm for computing the Donaldson-Thomas invariants, which correspond to BPS degeneracy of the corresponding BPS states in physics.  The main ingredients of this algorithm are the use of \textit{split attractor flows} and \textit{Kontsevich Soibelman wall crossing formulas}. Besides the known BPS spectrum in pure $SU(3)$ case, we obtain new family of BPS states with BPS-invariants equal to 2.
\end{abstract}

\tableofcontents

\section{Introduction}

The \textit{wall crossing structure} ("WCS" for short) is a formalism proposed and studied by M.Kontsevich and Y.Soibelman in \cite{kontsevich2013wall} that enables us to encode the Donaldson-Thomas (DT) invariants (BPS degeneracies in physics) and to control their "jumps" when certain walls (walls of marginal stability in physics) on the moduli space are being crossed. The celebrated Kontsevich-Soibelman wall crossing formula (''KSWCF" for short) \cite{kontsevich2008stability} is an essential ingredient of WCS. \\

WCS formalism is well adapted to the data coming from the complex integrable system. The famous Seiberg-Witten integrable system (``SW integrable system" for short) (\cite{donagi1997seiberg}) is an example. By considering certain gradient flows on the base of the integrable system called the \textit{split attractor flows}(c.f.\cite{moore1998arithmetic}\cite{denef2000supergravity}), WCS can produce an algorithm for computing the corresponding DT-invariants inductively.\\

This paper is about applying the WCS to the SW integrable systems associated to the pure $SU(2)$ and $SU(3)$ supersymmetric gauge theories. We will see that the results via the WCS formalism match perfectly well with those obtained via physics approaches. (\cite{fraser1997weak}\cite{taylor2001strong}). \\

The paper is organized as follows:\\

Section 2 gives a review of the formalism of wall crossing structures, which is based on the original paper of Kontsevich and Soibelman \cite{kontsevich2013wall}. The review is tailored toward the application in section 4, thus some aspects of this formalism (for example, its connection to variation of Hodge structures) are missing in this biased review. We end the review in subsection 2.5 by analyzing in detail two examples of WCS, namely, the \textit{stability data} and the parametrized family of stability data. The latter appears naturally in the framework of the complex integrable systems.\\

In section 3 we review the notion of complex integrable systems and its connection to the WCS formalism. We focus on the geometry of the base  $\mathcal{B}$, which can be endowed with the structure of \textit{special K\"ahler geometry}, as well as $\mathbb{Z}$-affine structure. Then we can consider \textit{split attractor flows} that are certain gradient flow lines on $\mathcal{B}$.  \\

Finally in section 4, we apply the WCS to SW integrable systems for $SU(2)$ and $SU(3)$. The main ingredients of the WCS for $SU(2)$ are the results of physics (see for example \cite{bergman1998strings}\cite{fayyazuddin1997results}\cite{bergman1998string}). We just reformulate in terms of WCS.\\

The results of subsection 4.3, in which we constructed WCS for pure $SU(3)$ case, are new. The construction roughly goes as follows:\\

We know from \cite{klemm1994monodromies}\cite{klemm1996nonperturbative}\cite{klemm1995simple} the vanishing cycles associated to the discriminant locus, which help producing the so-called \textit{initial condition} of the WCS. The WCS is then studied by cutting the real four dimensional base $\mathcal{B}$ with two complementary hyperplanes, so that the situation is reduced to the real two dimensional case. On the resulting plane after cutting, the discriminant locus was cut into six singular points that fall into three pairs, each of which determines a situation similar to the $SU(2)$ case considered in section 4.2. This corresponds to three embeddings of the Lie algebra $\mathfrak{su}(2)\hookrightarrow\mathfrak{su}(3)$. Thus, by applying the WCS formalism to this situation, we get three families of BPS states with non-trivial DT-invariants (or BPS invariants) (see proposition 4.9). \\

We apply again the WCS formalism to another wall (its existence is discussed for example in \cite{taylor2001strong}\cite{hollowood1997strong}) in order to derive the results about the interaction of the above three families. This gives us new BPS states with DT invariants equal to 2 (see proposition 4.11). \\

In \cite{galakhov2013wild}, it is exhibited by using the method of \textit{spectral networks} that even in this relatively simple $SU(3)$ case, the so called \textit{wild spectrum} exists. The wild spectrum is characterized by the \textit{exponential growth} of the DT-invariants. It would be very interesting to be able to produce the wild spectrum by using the WCS formalism considered in this paper. We leave the investigation along this direction for later works.

\newpage

\section{Wall Crossing Structures}

Fix a lattice $ \Gamma \cong {\mathbb{Z}}^{\oplus n} $ of finite rank, which is endowed with the pairing $\langle\cdot,\cdot\rangle:\wedge^2\,\Gamma\to\mathbb{Z}$. We call $\Gamma$ the \textit{charge lattice}. The \textit{central charge} is an abelian groups homomorphism $ Z: \Gamma\to \mathbb{C}$. Let $\mathfrak{g}:=\bigoplus_{\gamma\in\Gamma}\,\mathfrak{g}_{\gamma}$ be the $\Gamma$-graded Lie algebra over $\mathbb{Q}$ ($\mathfrak{g}_{\gamma}$ is generated over $\mathbb{Q}$ by $e_{\gamma}$), with Lie bracket:
\begin{equation}
    [e_{\gamma_{i}},e_{\gamma_{j}}]=(-1)^{\langle\gamma_{i},\gamma_{j}\rangle}\langle\gamma_{i},\gamma_{j}\rangle\cdot e_{\gamma_{i}+\gamma_{j}}
\end{equation}

The \textit{support} of $\mathfrak{g}$ is given by $ Supp\, \mathfrak{g}:=\{\gamma\in\Gamma:\mathfrak{g}_{\gamma}\ne 0\}\subset\Gamma$. We make the assumption that it is finite and contained in an open half-space in $\Gamma_{\mathbb{R}}:=\Gamma\otimes_{\mathbb{Z}}\mathbb{R}$. In this case, $\mathfrak{g}$ is \textit{nilpotent}. We denote by $G$ the corresponding nilpotent Lie group.

\subsection{Nilpotent case}

 The \textit{wall} associated to $\gamma$ is the hyperplane $\gamma^{\perp}\subset\Gamma_{\mathbb{R}}^{*}$. Denote by $Wall_{\mathfrak{g}}$ the finite collection of these hyperplanes. We borrow the following definition from \cite{kontsevich2013wall}:
\begin{definition}
 A \textbf{(global) \textit{wall crossing structure}} (WCS for short) for $\mathfrak{g}$ is an assignment $(y_{1},y_{2})\to g_{y_{1},y_{2}}\in G$ such that for any $y_{1},y_{2}\in\Gamma_{\mathbb{R}}^{*}\backslash Wall_{\mathfrak{g}}$ which is locally constant in $y_{1},y_{2}$, and satisfies the following\textbf{ \textit{cocycle condition}}
\begin{equation}
g_{y_{1},y_{2}}\cdot g_{y_{2},y_{3}}=g_{y_{1},y_{3}}\qquad\forall y_{1},y_{2},y_{3}\in \Gamma_{\mathbb{R}}^{*}-Wall_{\mathfrak{g}}
\end{equation}
and such that in the case when the straight interval connecting $y_{1}$ and $y_{2}$ intersects only one wall $\gamma^{\perp}$, then we have that
\begin{equation}
\log(g_{y_{1},y_{2}})\in\,\bigoplus_{{\gamma}^{\prime}\parallel\gamma}\mathfrak{g}_{{\gamma}^{\prime}}
\end{equation}
\end{definition}
Denote the space of all WCS on $\mathfrak{g}$ by $WCS_{\mathfrak{g}}$. Some easy consequences can be drawn from the above definition.

\begin{lemma}
If the straight interval $\overline{y_{1}y_{2}}$ does not intersect any walls in $Wall_{\mathfrak{g}}$, then $g_{y_{1},y_{2}}=id\in G$. And for any $y_{1},y_{2}\in \Gamma_{\mathbb{R}}^{*}\backslash Wall_{\mathfrak{g}}$, we have $g_{y_{1},y_{2}}=g_{y_{2},y_{1}}^{-1}$. More generally, for $y_{1},y_{2},\cdot\cdot\cdot,y_{n}\in\Gamma_{\mathbb{R}}^{*}\backslash Wall_{\mathfrak{g}}$, we have 
\begin{equation}
g_{y_{1},y_{2}}\cdot g_{y_{2},y_{3}}\cdot \cdot \cdot g_{y_{n-1},y_{n}}\cdot g_{y_{n},y_{1}}=id
\end{equation}
\end{lemma}

Thus with any codimension one stratum $\tau$, which is an open domain in $\gamma^{\perp}$ for some $\gamma\in Supp\,\mathfrak{g}$, one can associate a ``jump" $g_{\tau}:=g_{y_{1},y_{2}}$. Here points $y_{1}$ and $y_{2}$ are separated by the wall $\gamma^{\perp}$ and the straight segment $\overline{y_{1}y_{2}}$ intersects $\gamma^{\perp}$ at a point belonging to $\tau$. A WCS is uniquely determined by the collection of all jumps $(g_{\tau})_{codim\tau=1}$ that satisfies the cocycle condition for each stratum of codimension 2.

\subsection{Pronilpotent case}

Let $\Delta$ be strict sector (i.e., less than $180^{\circ}$) in $\Gamma_{\mathbb{R}}$, consider the pronilpotent Lie algebra $\mathfrak{g}_{\Delta}:=\prod_{\gamma\in\Delta\cap\Gamma\backslash\{0\}}\mathfrak{g}_{\gamma}$. There exists $\phi\in\Gamma_{\mathbb{R}}^{*}$ such that its restriction to $\Delta$ gives a proper map to $\mathbb{R}_{\geq 0}$. Then for $N>0$, consider the quotient $\mathfrak{g}_{\Delta,N}=\mathfrak{g}/\mathfrak{g}_{\Delta,\geq N}$, where $\mathfrak{g}_{\Delta,\geq N}\subset\mathfrak{g}_{\Delta}$ is the ideal consisting of elements with $\phi(\gamma)>N$ for $\gamma\in\Delta\cap\Gamma\backslash\{0\}$. We get $\mathfrak{g}_{\Delta}=\lim_{\stackrel{\longleftarrow}{N}}\mathfrak{g}_{\Delta,N}$. Then for each $N>0$, the WCS is defined on $\mathfrak{g}_{\Delta,N}$, and the WCS on $\mathfrak{g}$ corresponds to the one for the projective limit of $\mathfrak{g}_{\Delta,N}$. This can be done by using the sheaf theoretic description of WCS in subsection 2.4. Thus, for each $N>0$, we have a formula similar to (4), and by taking the limit $N\to\infty$, we obtain:
\begin{lemma}
  For any $y_{1},y_{2},\cdot\cdot\cdot,y_{n},\cdots\in\Gamma_{\mathbb{R}}^{*}- Wall_{\mathfrak{g}}$, we have the following \textbf{wall crossing formula}:
\begin{equation}
\left(\prod_{i=1}^{\infty}g_{y_{i},y_{i+1}}\right)\cdot g_{y_{\infty},y_{1}}=id
\end{equation}
\end{lemma}

The \textit{(numerical) Donaldson-Thomas invariants} proposed in \cite{kontsevich2008stability} gives us a map $\Omega: \Gamma\to \mathbb{Q}$ (conjectured to be integer-valued). Using \textit{KS-transform}:
$\mathcal{K}_{\gamma}:=exp\left(\sum_{n\geq 1}\frac{e_{n\gamma}}{n^{2}}\right)$, for the ray $\textit{l}_{\gamma}:=\mathbb{R}_{>0}\cdot Z(\gamma)$, we attach the element: \[\mathbb{S}(\textit{l}_{\gamma})=\prod_{Z(\gamma^{\prime})\in\textit{l}_{\gamma}}\mathcal{K}_{\gamma^{\prime}}^{\Omega(\gamma^{\prime})}=\prod_{\gamma^{\prime}\parallel\gamma}\mathcal{K}_{\gamma^{\prime}}^{\Omega(\gamma^{\prime})}=\prod_{\stackrel{Z(\mu)\in\textit{l}_{\gamma}}{m\geq 1}}\mathcal{K}_{m\mu}^{\Omega(m\mu)}\]

For strict sector $\Delta\subset\mathbb{C}^{*}\cong\mathbb{R}^{2}$ and fixed $N>0$, we have the truncation: $\mathbb{S}(\textit{l}_{\gamma})_{<N}\in G_{\Delta,N}$ (Lie group of $\mathfrak{g}_{\Delta,N}$), and thus the finite product $\mathbb{S}(\Delta)_{<N}:=\overrightarrow{\prod_{\textit{l}_{\gamma}\subset\Delta}}\,\mathbb{S}(\textit{l}_{\gamma})_{<N}\in G_{\Delta,N}$, where the right arrow indicates that the product is taken over all rays $\textit{l}_{\gamma}\subset\Delta$ in clockwise order. Let $N\to\infty$, we get: $\mathbb{S}(\Delta)=\overrightarrow{\prod_{\textit{l}_{\gamma}\subset\Delta}}\,\mathbb{S}(\textit{l}_{\gamma})\in G_{\Delta}$. Clearly, $\log\,\mathbb{S}(\textit{l}_{\gamma})\in\,\bigoplus_{{\gamma}^{\prime}\parallel\gamma}\mathfrak{g}_{{\gamma}^{\prime}}$. We identify the jump $g_{\tau}:=g_{y_{1},y_{2}}$ associated to $\gamma^{\perp}$ with $\mathbb{S}(\textit{l}_{\gamma})$.

\subsection{WCS on a topological space}

A WCS on $\mathfrak{g}$ assigns for any pair $y_{1},y_{2}\in\Gamma_{\mathbb{R}}^{*}\backslash Wall_{\mathfrak{g}}$ a group element $g_{y_{1},y_{2}}$. We want to assign group element to a single point $y\in\Gamma_{\mathbb{R}}^{*}\backslash\{0\}$. It was shown in \cite{kontsevich2013wall} (section 2.1.2) that there exists a \textbf{sheaf of wall crossing structures} over $\Gamma_{\mathbb{R}}^{*}$, denoted by $\mathcal{WCS}_{\mathfrak{g}}$, with the stalk over $y\in\Gamma_{\mathbb{R}}^{*}$ being $0$ if $y\in\Gamma_{\mathbb{R}}^{*}\backslash Wall_{\mathfrak{g}}$, and a nontrivial element $g_{y}$ such that $\log(g_{y})\in\,\bigoplus_{{\gamma}^{\prime}\parallel\gamma}\mathfrak{g}_{{\gamma}^{\prime}}$ if $y$ belongs to some wall. $WCS_{\mathfrak{g}}$ for $\mathfrak{g}$ with general $Supp\,\mathfrak{g}$ is defined as $ \mathcal{WCS}_{\mathfrak{g}_{\Delta}}:=\lim_{\stackrel{\longleftarrow}{N}}\mathcal{WCS}_{\mathfrak{g}_{\Delta,N}}$. 

Let $\mathcal{B}$ be a locally connected Hausdorff topological space. Given a locally continuous map $Y:\mathcal{B}\rightarrow\Gamma_{\mathbb{R}}^*$, 
let $\mathcal{WCS}_{\underline{\mathfrak{g}}, Y}$ be the pull back sheaf on $\mathcal{B}$, i.e., $\mathcal{WCS}_{\underline{\mathfrak{g}}, Y}:=Y^{*}(\mathcal{WCS}_{\mathfrak{g}})$.
\begin{definition}
\textbf{(WCS on $\mathcal{B}$)}: \textit{A (global) wall crossing structure on $\mathcal{B}$} is a global section of the sheaf $\mathcal{WCS}_{\underline{\mathfrak{g}}, Y}$.
\end{definition}

Given a section $\sigma\in\Gamma(U,\mathcal{WCS}_{\underline{\mathfrak{g}},Y})$ over open subset $U$, $ Supp\,\sigma$ (\textbf{support of $\sigma$}) is the minimal closed subset of $U\times \Gamma_{\mathbb{R}}$, conic in the direction of $\Gamma_{\mathbb{R}}$, that contains pairs $(b,\gamma)\in U\times \Gamma_{\mathbb{R}}$ such that $Y(b)(\gamma)=0$ and $\log\left(g_{(Y(b))}\right)_{}\in \mathfrak{g}_{\gamma}\backslash\{0\}.$

\begin{definition}
\textbf{(walls of second kind in $\mathcal{B}$)} For every $\gamma\in \Gamma\backslash\{0\}$, the wall of second kind associated to it is defined as the pull back of $\gamma^{\perp}$ through the map $Y$, that is: 
\begin{equation}
\mathcal{W}_{\gamma}^{2}:=\{b\in\mathcal{B}:\,Y(b)(\gamma)=0\}
\end{equation}
\end{definition}

\subsection{Wall Crossing Formulas}

The WCS on $\mathcal{B}$ is used to encode DT invariants $\Omega_{b}(\gamma)$ that depends on $b$. These invariants "jump" when certain walls on $\mathcal{B}$ are crossed. These walls are called the \textit{wall of the first kind} (\textit{walls of marginal stability} in physics literature).

\begin{definition}
For $\gamma_{1},\gamma_{2}\in\Gamma\backslash\{0\}$, the two $\mathbb{Q}$-linear independent elements, we introduce the following set
\begin{align}
    \mathcal{W}_{\gamma_{1},\gamma_{2}}:=\left \{b\in\mathcal{B}:\mathbb{R}_{>0}\cdot Z_{b}(\gamma_{1})=\mathbb{R}_{>0}\cdot Z_{b}(\gamma_{2})\right\}=\left\{b\in\mathcal{B}: Im\left(\frac{Z_{b}(\gamma_{1})}{Z_{b}(\gamma_{2}}\right)=0\right\}
\end{align}
Then the so-called \textbf{walls of the first kind} associated to a charge $\gamma\in\Gamma$ are defined as $\mathcal{W}_{\gamma}^{1}:=\bigcup_{\gamma=\gamma_{1}+\gamma_{2}} \mathcal{W}_{\gamma_{1},\gamma_{2}}$.
\end{definition}

The "jump" of DT invariants is controlled by the \textbf{wall crossing formula} of Kontsevich and Soibelman (\cite{kontsevich2008stability}): \textit{Given an strict sector $\Delta$, $\mathbb{S}(\Delta)_{b}$ does not depend on $b$ as long as there is no element $\gamma\in\,Supp\,\mathfrak{g}_{b}$ such that $Z_{b}(\gamma)$ crosses the boundary $\partial\Delta$, i.e., it is locally constant as a function in $b$ near $b_{0}$.} Thus, when approaching to $b_{0}$ on the side of $\mathcal{W}_{\gamma_{1},\gamma_{2}}$ where $\textit{l}_{\gamma_{1}}$ proceeds $\textit{l}_{\gamma_{2}}$ in the clockwise order, $\textit{l}_{\gamma_{1}},\textit{l}_{\gamma_{2}}$ would coalesce, and change the order after crossing the wall. But $\mathbb{S}(\Delta)$ stays constant by the above discussion. By taking $b\to b_{0}$ on both sides, the jump of $\Omega_{b}(\gamma)$ can then be determined. 

Consider the sub-lattice $\Gamma_{0}\subset\Gamma$ generated by the positive cone $\mathcal{C}_{0}:=\mathbb{Z}_{\geq0}\cdot\gamma_{1}\oplus\mathbb{Z}_{\geq0}\cdot\gamma_{2}=\{m\gamma_{1}+n\gamma_{2}:m,n\in\mathbb{Z}_{\geq0}\}$. Write $\gamma=m\gamma_{1}+n\gamma_{2}$ simply as $(m,n)$, the KSWCF then states that $\prod_{\textit{l}_{\gamma}\subset\Delta_{\gamma_{1},\gamma_{2}}}^{\longrightarrow}\mathbb{S}_{b_{-}}(\textit{l}_{\gamma})=\prod_{\textit{l}_{\gamma}\subset\Delta_{\gamma_{1},\gamma_{2}}}^{\longrightarrow}\mathbb{S}_{b_{+}}(\textit{l}_{\gamma})$, which by using KS-transformation, it reads as
\begin{equation}
\prod_{m,n\geq0;\,m/n\nearrow}\mathcal{K}_{(m,n)}^{\Omega_{b_{-}}(m,n)}=\prod_{m,n\geq0;\,m/n\searrow}\mathcal{K}_{(m,n)}^{\Omega_{b_{+}}(m,n)}
\end{equation}
Notice that the product on LHS (RHS) is taken over all coprime $m,n$ in the increasing (decreasing) order of $m/n\in\mathbb{Q}$.
 
\noindent\textbf{Examples}: The form of KSWCF depends on $k=\langle\gamma_{1},\gamma_{2}\rangle$. $k=1$ gives the \textit{pentagon identity}
\begin{equation}
    \mathcal{K}_{\gamma_{1}}\mathcal{K}_{\gamma_{2}}=\mathcal{K}_{\gamma_{2}}\mathcal{K}_{\gamma_{1}+\gamma_{2}}\mathcal{K}_{\gamma_{1}}
\end{equation}
$k=2$ gives the following which is relevant to the representation of the Kronecker quiver $K_{2}:=\{\bullet\rightrightarrows\bullet\}$:
\begin{equation}
    \mathcal{K}_{\gamma_{2}}\,\mathcal{K}_{\gamma_1}=\left(\prod_{n=1}^{\infty}\mathcal{K}_{n\gamma_{1}+(n-1)\gamma_{2}}\right)\,\mathcal{K}_{\gamma_{1}+\gamma_{2}}^{-2}\,\left(\prod_{n=\infty}^{1}\mathcal{K}_{(n-1)\gamma_{1}+n\gamma_{2}}\right)
\end{equation}
Finally, we point out that the formula (8) is equivalent to the WCF (5). Indeed, consider a loop around $b\in\mathcal{W}_{\gamma_{1},\gamma_{2}}$ that intersects countably many walls of second kind $\mathcal{W}_{m\gamma_{1}+n\gamma_{2}}^{2}\, \text{for}\,\,(m,n)\in\mathbb{Z}_{\geq0}\times\mathbb{Z}_{\geq0}\,\,\text{and}\,\, m+n>0$. Then (8) can be informally interpreted as the triviality of the monodromy of $\mathbb{S}(\textit{l}_{\gamma})$ along the small loop, i.e.,$\stackrel{\circlearrowright}{\underset{m/n\nearrow}{{\prod}}}\mathbb{S}(\textit{l}_{m\gamma_{1}+n\gamma_{2}})=id$.

\subsection{Examples of WCS}

A \textbf{stability data} on $\mathfrak{g}$ (see \cite{kontsevich2008stability}) is a pair $\sigma=(Z,a)$ consisting of a central charge and a collection of elements $a(\gamma)\in\mathfrak{g}_{\gamma}$ that satisfies the \textit{support property} (we omit it here). Denote by $Stab(\mathfrak{g})$ the space of all stability data on $\mathfrak{g}$. Since $\mathbb{S}(\textit{l}_{\gamma})=\exp\left\{\sum_{Z(\gamma)\in\textit{l}_{\gamma}}a(\gamma)\right\}$, where $a(\gamma)=\sum_{k|\gamma,k\ge1}\frac{\Omega(\gamma/k)}{k^2}e_{\gamma}\in\mathfrak{g}_{\gamma}$, $Stab(\mathfrak{g})$ is the same as the set $\widehat{Stab}(\mathfrak{g})$ which consists of pairs $(Z,\mathbb{S})$, where $\mathbb{S}$ a collection of elements $(\mathbb{S}_{\Delta})_{\Delta\in\mathcal{S}}$, with $\mathbb{S}_{\Delta}\in G_{\Delta}$. 

\textbf{Example a)}: Take $\mathcal{B}=S_{\theta}=\mathbb{R}/2\pi\mathbb{Z}$ and consider $Y: S_{\theta} \to \Gamma_{\mathbb{R}}^{*},\,\,\theta\mapsto Y_{\theta}(\gamma):=Im(e^{-i\theta}Z(\gamma))$, $\theta\in\mathbb{R}/2\pi\mathbb{Z},\gamma\in\Gamma$.
 
\begin{proposition}
 The WCS on $\mathbb{R}/2\pi\mathbb{Z}$ induced by $Y$ above is the same as a stability data.
\end{proposition}
\begin{proof}
  Given $s\in\Gamma(\mathcal{WCF}_{\mathfrak{g},Y})$, recall that $\left(\mathcal{WCS}_{\underline{\mathfrak{g}},Y}\right)_{\theta}=\left(\mathcal{WCS}_{\underline{\mathfrak{g}}}\right)_{Y(\theta)}$. Thus, if $Y(\theta)\in\gamma^{\perp}$, i.e., $Y(\theta)(\gamma)=0$, we can associate to $\gamma$ a ``jump" $\mathbb{S}(\textit{l}_{\gamma})\in G$ such that $\log\,\mathbb{S}(\textit{l}_{\gamma})\in\bigoplus_{\gamma^{\prime}\parallel\gamma}\mathfrak{g_{\gamma^{\prime}}}$, where $\mathbb{S}(\textit{l}_{\gamma})$ is the group element associated to $\textit{l}_{\gamma}=Z(\gamma)\cdot\mathbb{R}_{>0}=e^{i\theta}\cdot\mathbb{R}_{>0}$. Then, the components of $\log\,\mathbb{S}(\textit{l}_{\gamma})$ belonging to $g_{\gamma}$ is deemed as $a(\gamma)$. Conversely, given $(Z,a(\gamma))\in Stab\,(\mathfrak{g})$, it was shown in \cite{kontsevich2013wall} (see section 2.1.2) that there exists element $g_{+,-}\in G$ that determines uniquely a section of $\mathcal{WCS}_{\mathfrak{g}}$, which after pulling back through $Y$, gives arise to a section of $\mathcal{WCS}_{\underline{\mathfrak{g}},Y}$. 
\end{proof}
In order to recover $Stab\,(\mathfrak{g})$, we consider the WCS on the space $\widehat{\mathcal{B}}:=S_{\theta}\times Hom_{\mathbb{Z}}(\,\Gamma,\mathbb{C})$ via the map
\[Y:S_{\theta}\times Hom_{\mathbb{Z}}(\,\Gamma,\mathbb{C})\to\Gamma_{\mathbb{R}}^{*}\qquad (\theta,Z)\mapsto Y(\theta,Z)(\gamma):=Im(e^{-i\theta}Z(\gamma)),\,\,\forall\gamma\in\Gamma\backslash\{0\}\]
\begin{proposition}
 A section of the sheaf $\mathcal{WCS}_{\underline{\mathfrak{g}},Y}$ on the space  $\widehat{\mathcal{B}}$ is same as a family of stability data on $\mathfrak{g}$.
\end{proposition}
\begin{proof}
 For fixed $Z\in Hom_{\mathbb{Z}}(\,\Gamma,\mathbb{C})$, by restricting $\mathcal{WCS}_{\underline{\mathfrak{g}},Y}$ to $\mathcal{B}_{Z}=S_{\theta}\times\{Z\}\subset\widehat{\mathcal{B}}$, we get a stability data by above proposition. Moving $Z$ locally, we get a germ of universal family of stability data with central charges near $Z$.
\end{proof}
\textbf{Example b): Continuous family of stability data recovered}

In application, we usually use a complex manifold $\mathcal{B}$ to parametrize $Stab\,(\mathfrak{g})$, which can be realized as the WCS on $\mathcal{B}_{\theta}:=\mathcal{B}\times S_{\theta}$ through the map
\begin{equation}
Y:\mathcal{B}\times S_{\theta}\to\Gamma_{\mathbb{R}}^*\quad (b,\theta)\mapsto Y_{\theta}(b)(\gamma):=Im(e^{-i\theta}Z_{b}(\gamma))\,\,\,\forall\gamma\in\Gamma\backslash\{0\}
\end{equation}

\begin{proposition}
  A section of the sheaf $\mathcal{WCS}_{\underline{\mathfrak{g}},Y}$ on $\mathcal{B}_{\theta}$ is same as a WCS on $\mathcal{B}$.
\end{proposition}

\begin{proof}
   Fix $b_{0}\in\mathcal{B}$, we get a stability data on $\mathfrak{g}_{b_{0}}$. As $b$ varies locally near $b_{0}$, $Z_{b}$ also varies locally around $Z_{b_{0}}$ (as $Z$ depends holomorphically on $b$). Thus by proposition 2.4, we get a germ of universal family of stability data near $b_{0}$. But this is the same as a germ of universal family of WCS near $b_{0}$.
\end{proof}

Given a WCS on $\mathcal{B}_{\theta}$, we get a locally constant map $a:tot(\underline{\Gamma})\to tot(\underline{\mathfrak{g}}),\,(b,\gamma)\mapsto a_{b}(\gamma)\in\underline{\mathfrak{g}}_{b,\gamma}$, where $a_{b}(\gamma)$ is the $\gamma$-component of the corresponding section of $\mathcal{WCS}_{\underline{\mathfrak{g}},Y}$. Then $a_{b}(\gamma)$ is non-trivial only if there exists $\theta\in S_{\theta}$ such that $Y_{\theta}(b)(\gamma)=0$. Let $\mathcal{B}^{\prime}:=\{(b,\gamma)\in tot\,(\underline{\Gamma}):\exists\,\theta  (\text{fixed}),\,Y_{\theta}(b)(\gamma)=0\}\subset tot\,(\underline{\Gamma})$, we see that $a$ restricts to $\mathcal{B}^{\prime}$.

\begin{remark}
Since the map $a:\mathcal{B}^{\prime}\to tot\,(\underline{\mathfrak{g}})$ depends only on the $b$-component of $\mathcal{B}_{\theta}$, so we will just call a WCS on $\mathcal{B}_{\theta}$ as a WCS on $\mathcal{B}$, with the role of the circle $S_{\theta}$ being understood implicitly.
\end{remark}

 Recall from the definition 2.3, given $\gamma\in\Gamma$, the wall of second kind associated to it is given by\[\mathcal{W}_{\gamma}^{2}:=\gamma^{\perp}=\{(b,\theta)\in\mathcal{B}_{\theta}:\,Y_{\theta}(b)(\gamma)=0\}=\{(b,\theta)\in\mathcal{B}_{\theta}:\,Im(e^{-i\theta}Z_{b}(\gamma)=0\}\]

which is seen to be a hypersurface in $\mathcal{B}_{\theta}$. And the pull back of the wall of second kind in $\Gamma_{\mathbb{R}}^{*}$ through the map $Y$. We can project this wall down to a hypersurface in $\mathcal{B}$, also denoted by $\mathcal{W}_{\gamma}^{2}$, i.e.,
\begin{equation}
\mathcal{W}_{\gamma}^{2}=\{b\in\mathcal{B}:\,\exists\,\theta\,(\text{fixed})\,s.t.,\,\,Arg(Z_{b}(\gamma))=\theta\}\subset\mathcal{B}
\end{equation}

\section{Complex integrable systems and WCS}

In this section, we review the notion of complex integrable systems following the treatment in \cite{kontsevich2013wall} with emphasis on the connection to WCS. We will see that a WCS can be induced from the data of the complex integrable system. This produces an algorithm for computing DT-invariants by employing certain trees on its base given in terms of the split attractor flows.

\subsection{Complex integrable system and its geometry}

A \textit{complex integrable system} (see section 4 of \cite{kontsevich2013wall}) is a holomorphic surjective map $\pi:(X,\omega^{2,0})\to \mathcal{B}$ with smooth fibers being holomorphic Lagrangian submanifolds, where $(X,\omega^{2,0})$ is a holomorphic symplectic manifold of complex dimension $2n$, and the base $\mathcal{B}$ is a complex manifold of dimension $n$. Denote by $\mathcal{B}^{0}\subset\mathcal{B}$ the dense open subset over which the fibers being smooth, and by $\Delta=\mathcal{B}^{sing}:=\mathcal{B}-\mathcal{B}^0$ the \textit{discriminant locus}.  We make the assumption that the smooth fibres are compact, so they are actually holomorphic Lagrangian tori. It is well-known that:  
\begin{proposition}
 The base of the complex integrable system $\pi:(X^{0},\omega^{2,0})\to\mathcal{B}^{0}$ is a $\mathbb{Z}$-affine manifold.
\end{proposition}

\begin{definition}
We have a local system of lattices $\underline{\Gamma}$ over $\mathcal{B}^{0}$, with its ``stalk" $\underline{\Gamma}_{b}$ at $b\in\mathcal{B}^{0}$ being the first homology group $H_{1}(\pi^{-1}(b),\mathbb{Z})$. $\underline{\Gamma}$ will be called the local system of \textbf{charge lattices}.
\end{definition}
The one forms $\alpha_{i}:=\oint_{\gamma_{i}}\omega^{2,0},\,\,i=1,\cdots,2n$ are closed, So $\exists$ locally holomorphic functions $z_{i}$ s.t. $dz_{i}=\alpha_{i}$, which define an holomorphic map $(z_{1},\cdots,z_{2n}): U\to\mathbb{C}^{2n}$ for $U$ an open neighborhood of $b$. 
\begin{definition}
(\cite{kontsevich2013wall}) The collection of one forms $\alpha_{i}$ gives rise to an element $\delta\in H^{1}(\mathcal{B}^{0},\underline{\Gamma}^{\vee}\otimes\mathbb{C})$. If $\delta$ is zero, there exists a section $Z\in\Gamma(\mathcal{B}^{0},\underline{\Gamma}\otimes\mathscr{O}_{\mathcal{B}^{0}})$, called \textbf{central charge} such that $dZ(\gamma_{i})=dz_{i}=\alpha_{i}$\, for $1\leq i\leq2n$.
\end{definition}

If the fibers of $\pi$ are endowed with a covariantly constant integer polarization. We call such system the \textit{polarized complex integrable system}. A polarization gives rise to $\langle\cdot,\cdot\rangle: \Lambda^{2}\,\underline{\Gamma}\to\underline{\mathbb{Z}}\,_{\mathcal{B}^{0}}$. Denote by $(\omega_{ij}):=(\langle\gamma_{i},\gamma_{j}\rangle)$ the corresponding matrix and by $(\omega^{ij})$ its inverse. We have that (section 4.1.1 of \cite{kontsevich2013wall}):

\begin{proposition}
  a) $\sum_{i,j}\omega^{ij}dz_{i}\wedge dz_{j}=\sum_{i,j}\omega^{ij}\alpha_{i}\wedge\alpha_{j}=0$\[b)\,\sqrt{-1}\sum_{i,j}\omega^{ij}dz_{i}\wedge d\bar{z}_{j}=\sqrt{-1}\sum_{i,j}\omega^{ij}\alpha_{i}\wedge\bar{\alpha}_{j} >0.\]
\end{proposition}
In terms of the central charge $Z$, the condition a) and b) above can be rewritten as
\begin{equation*}
a)^{\prime} \,\,\,\langle dZ,dZ\rangle=0\,\,\,(\text{transversality})\qquad b)^{\prime}\,\,\, \sqrt{-1}\,\langle dZ,d\bar{Z}\rangle>0\,\,\,(\text{nondegeneracy})
\end{equation*}

Choose \textit{symplectic basis} $\{\alpha^{i},\beta_{j}\}$ for $\underline{\Gamma}$ near $b\in\mathcal{B}^{0}$. Then $Z=\sum_{i=1}^{n}a^{i}\alpha_{i}+\sum_{i=1}^{n}a_{D,i}\beta^{i}$, where $a^{i}$, $a_{D,i}\in\mathscr{O}_{\mathcal{B}^{0}}$ are locally holomorphic on $\mathcal{B}^{0}$, called the \textit{special coordinates} (see \cite{freed1999special} for its definition). Then condition $a^{\prime})$ specialize into $d\left(\sum_{i}a_{D,i}\,da^{i}\right)=0$, from which we deduce that $\sum_{i}a_{D,i}\,da^{i}=d\mathcal{F}$ for some holomorphic function $\mathcal{F}$, called the \textit{prepotential}. And we have: $a_{D,i}=\frac{\partial\mathcal{F}}{\partial a^{i}},\,\,i=1,\cdots,n$.

Condition $b^{\prime})$ becomes $Im\left(\sum_{i}d{a}_{D,i}\wedge d\bar{a}^{i}\right)>0$, thus we have the following K\"ahler metric on $\mathcal{B}^{0}$:
\begin{equation}
    g_{\mathcal{B}^{0}}=Im\left(\sum_{i}d{a}_{D,i}\, d\bar{a}^{i}\right)=\sum_{i,j}\,Im(\tau_{ij})da^{i}d\bar{a}^{j}
\end{equation}
where the matrix $\tau=(\tau_{ij})$ is given by $\tau_{ij}:=\frac{\partial a_{D,i}}{\partial a^{j}}=\frac{\partial^{2}\mathcal{F}}{\partial a^{i}\partial a^{j}}$.

\begin{remark}
 Define $x^{i}:=Re(a^{i})$ and $y_{j}:=Re(a_{D,j})$, then $\{x^{i},y_{j}\}$ are $\mathbb{Z}$-affine coordinates on $\mathcal{B}^{0}$ (see proposition 3.1). We can also use $\{Im\, a^{i},Im\,a_{D,j}\}$ as $\mathbb{Z}$-affine coordinates (\textit{dual affine structure}). They differ by a $\frac{\pi}{2}$-rotation of the central charge $Z$. More generally, we can consider the rotated $\mathbb{Z}$-affine coordinates: $\{Re(e^{-i\theta}a^{i})\}$ or $\{Im(e^{-i\theta}a_{D,i})\}$.
\end{remark}

 Near $\Delta$, the fibration $\pi$ degenerates. Some cycles, which we call the \textit{vanishing cycles} shrink to zero. We give the following\textbf{ $A_{1}$-singularity assumption} (section 4.5 of \cite{kontsevich2013wall}). 
 
 Assume that $\Delta=\mathcal{B}\backslash\mathcal{B}^{0}$ is an analytic divisor, and there exists an analytic divisor $\Delta^{1}\subset\Delta$ such that $dim \Delta^{1}\leq dim\mathcal{B}^{0}-2$, and the complement $\Delta^{0}:=\Delta\backslash\Delta^{1}$ is smooth, then we have:
\begin{enumerate}
    \item There exist local coordinates $\{z_{1},\cdots,z_{n}\}$ near $\Delta^{0}$ such that $\Delta^{0}=\{z_{1}=0\}$.
    \item $Z:\mathcal{B}^{0}\to\underline{\Gamma}^{\vee}\otimes_{\mathbb{Z}}\mathbb{C}\cong\mathbb{C}^{2n}$ is a multi-valued and is given by $(z_{1},\cdots,z_{n})\mapsto(z_{1},\cdots,z_{n},\partial_{1}F_{0},\cdots,\partial_{n}F_{0})$, where $\partial_{i}:=\frac{\partial}{\partial z_{i}}$, and $F_{0}=\frac{1}{2\pi i}\frac{z_{1}^{2}}{2}\log z_{1}+G(z_{1},\cdots,z_{n})$, where $G$ is holomorphic, and the \textit{prepotential} $F_{0}$ satisfies the \textit{positivity condition}: $i\langle dZ,\overline{dZ}\rangle>0$.
    \item The monodromy of $\underline{\Gamma}$ about $\Delta_{0}$ is given by the \textit{Picard-Lefshetz type formula}: $\gamma\mapsto\gamma+\langle\gamma,\gamma_{0}\rangle\gamma_{0}$, where $\gamma_{0}$ is the vector such that $\langle\gamma_{0},\cdot\rangle\in\underline{\Gamma}^{\vee}$ is a primitive covector. 
\end{enumerate}

\subsection{Attractor flow and its properties}

We introduce attractor flows on the base. They first appeared in the study of super-gravity (\cite{moore1998arithmetic}\cite{denef2000supergravity}).  Given $\gamma\in\underline{\Gamma}$ near $b_{0}\in\mathcal{B}^{0}$, consider the function $F_{\gamma}(\mathbf{u}):=Re(Z_{\mathbf{u}}(\gamma))$, where $\mathbf{u}$ is the complex coordinate on $\mathcal{B}^{0}$.

\begin{definition}
The \textbf{attractor flow} associated to $(b_{0},\gamma)\in tot(\underline{\Gamma})$ is given by the gradient flow of the function $F_{\gamma}(\mathbf{u})$, namely $\dot{\mathbf{u}}+\nabla F_{\gamma}(\mathbf{u})=0$, where $\dot{\mathbf{u}}$ denotes the derivative with respect to the ``time" parameter $t$, and the gradient is taken with respect to the K\"ahler metric (13), i.e., its $i$-th component is given by $\sum_{j}g^{i\bar{j}}\,\overline{\partial}_{j}F_{\gamma}(\mathbf{u})$.
\end{definition}

 As $Z({\gamma})$ is a holomorphic, $F_{\gamma}(\mathbf{u})$ is decreasing and $Im(Z({\gamma}))$ stays constant along the flow lines. In the affine structure corresponding to $\theta=Arg Z_{b_{0}}(\gamma)$, the flow line equation takes the form: $Im\left(e^{-i\theta}Z_{b(t)}(\gamma)\right)=0$.
\begin{proposition}
The phase $Arg Z(\gamma)$ of the central charge is constant along the attractor flow. \end{proposition}

\begin{definition}
Since $F_{\gamma}(\mathbf{u})$ is decreasing along the flow, the attractor flow would converge to the local minima of the function $F_{\gamma}(\mathbf{u})$, these terminate points of the flow line will be called the \textbf{attractor points}.
\end{definition}

\begin{remark}
Since $Im\left(e^{-i\theta}Z_{b}(\gamma)\right)$ vanishes identically along the flow line, $|Z_{b}(\gamma)|=Re\left(e^{-i\theta}Z_{b}(\gamma)\right)$. This suggest that the attractor points correspond to the minimum of the function $|Z_{}(\gamma)|$, which has the meaning as ``mass" for BPS particles in physics. Thus, we call $Re\left(e^{-i\theta}Z_{b}(\gamma)\right)$ the \textbf{mass function}.
\end{remark}

The possible minima of $|Z_{b}(\gamma)|$ are the zeros of $Z(\gamma)$. The vanishing cycle $\gamma$ causes $Z_{b}(\gamma)\to 0\, \text{as}\, b\to\Delta$. Thus we expect that there are attractor points belonging to $\Delta$. We need to study the behavior of attractor flows near $\Delta$. By $A_{1}$-singularity assumption, $\Delta=\{z_{1}=0\}$, where $\{z_{1},\cdots,z_{n}\}$ are local coordinates. Reducing to the two dimensional case, denote by $\gamma_{0}$ the vanishing cycle and $\gamma_{1}$ the remaining basis element, we see that near $\Delta$, we have that \[Z_{u}(\gamma_{0})=u,\,\,\,Z_{u}(\gamma_{1})=\frac{1}{2\pi i}\left(u\log u+\frac{u}{2}\right)+\partial_{u}G(u).\] 

The cycle $\gamma_{0}$, under $T\mathcal{B}^{0}\cong \underline{\Gamma}$, corresponds to the \textit{invariant direction} under the Picard-Lefshetz monodromy.
\begin{proposition}
Under $A_{1}$-singularity assumption, the flow line associated to vanishing cycle $\gamma_{0}$ terminates at the singularity (thus an attractor point). For charges other than $\gamma_{0}$, the flow line can avoid this singularity.  
\end{proposition}
\begin{proof}
$Z_{u}(\gamma_{0})=u$ goes to zero as $u\to0$, which is the minimum of $Re\left(Z_{u}(\gamma_{0})\right)$, so the flow line terminates at the singularity $\{u=0\}$. On the other hand, if the flow associated to $Re\left(Z_{u}(\gamma_{1})\right)$ terminates at the singularity, then by proposition 3.3, $Arg(Z(\gamma_{0}))$ keeps constant, but this is impossible by the particular form of $Z_{u}(\gamma_{1})$ given above. 
\end{proof}
\textbf{Split attractor flows}: If the flow associated to $\gamma$ hits the wall $\mathcal{W}_{\gamma_{1},\gamma_{2}}$ where $Arg (Z(\gamma))=Arg (Z(\gamma_{1}))=Arg (Z(\gamma_{2}))$, as $Arg (Z(\gamma))$ stays constant along the flow, the flow splits into two flow lines associated to $\gamma_{1}$ and $\gamma_{2}$ at this split point. Denote by $\mathcal{L}_{\gamma}$ the flow of $\gamma$, thus we have schematically: $\mathcal{L}_{\gamma}=\mathcal{L}_{\gamma_{1}}+\mathcal{L}_{\gamma_{2}}$. Clearly, the process can be iterated until the resulting flow lines all terminate at the attractor points. This gives a tree on $\mathcal{B}^{0}$, called the \textbf{attractor tree}.

\textbf{Connection to WCS}: To produce a collection of numbers $\Omega_{b}(\gamma)$ for $(b,\gamma)\in tot\,\underline{\Gamma}$ satisfying KSWCF, we start with the \textbf{``initial data"}(i.e., DT-invariants of vanishing cycles at $\Delta$ equals 1). And consider all attractor trees rooted at $b$ that terminate at $\Delta$. Assuming their number being finite, they form a graph without oriented cycles. Then move toward the root, and apply KSWCF at each internal vertex $b_{*}$ belonging to the first kind walls, we will arrive at $\Omega_{b}(\gamma)$ inductively. We prove the following result due to Kontsevich and Soibelman (see section 2.7 of \cite{kontsevich2008stability})

\begin{proposition}
The WCS on $\mathcal{B}$ gives rise to a local embedding $\mathcal{B}^{0}\hookrightarrow Stab(\mathfrak{g}_{b})$ for each $b\in\mathcal{B}^{0}$.
\end{proposition}
\begin{proof}
As $b$ varies, we get a family of stability data on $\mathfrak{g}_{b}$ identified with the WCS on $\mathcal{B}_{\theta}:=\mathcal{B}\times S_{\theta}$ via (11). Given $(b,\gamma)$, consider \textbf{all} good attractor trees rooted at $b$. As the phase of $Z_{b}$ stays constant along the flow, $a$ or $\Omega:tot\,\underline{\Gamma}\to\underline{\mathbb{Q}}$ restricts to $\mathcal{B}^{0\prime}:=\{(b,\gamma):Y_{\theta}(b)(\gamma)=0\}$. Thus the flow lines sit in the wall $\mathcal{W}^{2}_{\gamma}$ (see definition 2.3). By the procedure above, we obtain a collection $\Omega_{b}(\gamma)$ that satisfy KSWCF. The attractor trees on $\mathcal{B}$ can be lifted to trees on $\mathcal{B}\times S_{\theta}$. In particularly, for $b\in\mathcal{B}^{0}$, we have a WCS on $\{b\}\times S_{\theta}\cong S_{\theta}$, which is the same as a stability data on $\mathfrak{g}_{b}$. 
\end{proof}

\section{WCS in Seiberg-Witten integrable systems}

 We review Seiberg-Witten integrable systems associated to $SU(n)$ (\cite{donagi1997seiberg}\cite{witten1996supersymmetric}) which originated from the study of supersymmetric Yang-Mills theories (\cite{lerche1998introduction} \cite{seiberg1994electric}). We will then apply the WCS formalism to the $n=2$ and $n=3$ cases.

\subsection{Seiberg-Witten integrable systems for SU(n)}

Let $\mathcal{W}_{A_{n-1}}(x,\mathbf{u})\equiv x^{n}-\sum_{k=0}^{n-2}u_{k+2}(\mathbf{u})\,x^{n-2-k}$ be the simple singularities associated to $SU(n)$. The SW curve is then $\mathcal{C}_{\mathbf{u}}: y^{2}=P_{n}(x)=\left(\mathcal{W}_{A_{n-1}}(x,\mathbf{u})\right)^{2}-\Lambda^{2n},\, x,y\in\mathbb{C}$, and $\Lambda$ is a real parameter. The discriminant locus $\Delta_{\Lambda}$ is the variety of the discriminant $\delta_{\lambda}$ of the polynomial $P_{n}(x)$. 

The \textbf{SW integrable system} $\pi:X\to\mathcal{B}$ has the base $\mathcal{B}=\mathbb{C}^{n-1}$, with the fiber over $b\in\mathcal{B}^{0}$ being the Jacobian $Jac(\mathcal{C}_{b})$. The charge lattices $\underline{\Gamma}$ are formed by $\Gamma_{b}:=H_{1}(\mathcal{C}_{b},\mathbb{Z})$, endowed with the constantly covariant intersection parings $\langle\cdot,\cdot\rangle: \Lambda^{2}\,\underline{\Gamma}\to\underline{\mathbb{Z}}_{\mathcal{B}^{0}}$. Let $\{\omega_{1},\cdots,\omega_{g}\}$ be a basis of $H^{0}(\mathcal{C}_{b},\Omega_{\mathcal{C}_{b}}^{1})$, and $\left\{\alpha^{i},\beta_{j}\right\}_{1\leq i,j\leq g}$ be a symplectic basis of $\underline{\Gamma}_{b}$ with the intersection matrix $\mathbf{J}=\left(\begin{array}{cc}
          \mathbf{0} & \mathbf{I}_{g} \\
         -\mathbf{I}_{g} & \mathbf{0} 
    \end{array}\right)$. The \textbf{period matrix} $\mathbf{\Omega}=(\mathbf{A}\,\,\mathbf{B})$, in which $\mathbf{A}$ and $\mathbf{B}$ are  matrix with entries given respectively by $\mathbf{A}_{ij}:=\oint_{\alpha^{i}}\omega_{j}\,\,\mathbf{B}_{ij}:=\oint_{\beta_{i}}\omega_{j}$, satisfies the following relations:
\begin{equation*}
  \text{Riemann $1^{st}$ bilinear relation}\,\,\,\mathbf{\Omega}\mathbf{J}^{t}\mathbf{\Omega}=0
  \end{equation*}
  \begin{equation}
  \text{Riemann $2^{nd}$ bilinear relation}\,\,\,\sqrt{-1}\,\mathbf{\Omega}\mathbf{J}^{t}\bar{\mathbf{\Omega}}>0
 \end{equation}

\begin{definition}
Define the $\tau$-matrix to be $\tau:=\mathbf{A}^{-1}\mathbf{B}$. Riemann's $2^{nd}$ bilinear relation implies that $\tau$ is symmetric and its imaginary part is positive definite, i.e., $Im\tau>0$, which can be used to define the K\"ahler metric (see (13)) on $\mathcal{B}^{0}$.
\end{definition} 
    
\begin{proposition}
The torus fibration $\pi:X\rightarrow\mathcal{B}$ defines a complex integrable system with central charge.
\end{proposition}
 
\begin{proof}
The smooth fiber, being the Jacobian torus, is endowed with the angle coordinates $\left\{\theta^{i},\theta_{j}\right\}_{1\leq i,j\leq g}$ such that $\oint_{\alpha^{i}}d\theta_{j}=\delta_{j}^{i}$\, and
    $\oint_{\beta_{i}}d\theta^{j}=\delta_{j}^{i}$. Let $w_{i}:=\theta_{i}+\sum_{j}\tau_{ij}\,\theta^{j}$ be the complex coordinates along the fibers. Together with the coordinates $a^{i}$ on $\mathcal{B}$, the holomorphic symplectic form can be written as: $ \omega^{2,0}=\sum_{i}da^{i}\wedge dw_{i}$, which vanishes when restricted to generic fibers. By definition 3.2, the central charge $Z$ is then given through $dZ(\gamma)=\oint_{\gamma}\,\omega^{2,0}$.
\end{proof}

Recall that we have defined (section 3.1) the symmetric matrix $\tau=(\tau_{ij})=\left(\frac{\partial a_{D,i}}{\partial a^{j}}\right)=\left(\frac{\partial^{2}\mathcal{F}}{\partial a^{i}\partial a^{j}}\right)$, then we have

\begin{proposition}
The matrix above can be identified with that defined in definition 4.1 via the period matrix $\mathbf{\Omega}$.
\end{proposition}

\begin{proof}
Using $\alpha_{\text{\,can}}=\sum_{i}a_{D,i}\wedge d\theta^{i}+a^{i}\wedge d\theta_{i}$ so that $d\alpha_{\text{\,can}}=\omega^{2,0}$, we write  $a^{i}=\oint_{\alpha^{i}}\, \alpha_{\text{\,can}},\,\,
      a_{D,j}=\oint_{\beta_{j}}\,\alpha_{\text{\,can}}$. We can construct a meromorphic form $\lambda_{SW}$ (\textbf{Seiberg-Witten differential}) with vanishing residues (lemma 4.3 below) in the class of $[\alpha_{\text{\,can}}]$ so that $a^{i}=\oint_{\alpha^{i}}\, \lambda_{SW},\,\,a_{D,j}=\oint_{\beta_{j}}\,\lambda_{SW}$. Then we have \[\tau_{ij}=\frac{da_{D,i}}{da^{j}}=\sum_{k}\frac{\partial a_{D,i}}{\partial u_{k}}\Bigm/\frac{\partial a^{j}}{\partial u_{k}}=\sum_{k}\oint_{\beta_{i}}\frac{\partial\lambda_{SW}}{\partial u_{k}}\Bigm/\oint_{\alpha^{i}}\frac{\partial\lambda_{SW}}{\partial u_{k}}\] 
      
      Thus, for $\tau$ matrix to match that defined via $\mathbf{\Omega}$, we need $\frac{\partial\lambda_{SW}}{\partial u_{k}}=f_{k}(\mathbf{u})\,\omega_{k}$ for some holomorphic functions $f_{k}(\mathbf{u})$, which means that $\left\{\frac{\partial\lambda_{SW}}{\partial u_{k}}\right\}_{1\leq k\leq g}$ form a basis of $H^{0}\left(\mathcal{C}_{\mathbf{u}},\Omega_{\mathcal{C}_{u}}^{1}\right)$ up to a scalar multiplication induced by $f_{k}(\mathbf{\mathbf{u}})$. 
      \end{proof}
      We see that the central charge can be given as $Z(\gamma)=\oint_{\gamma}\,\lambda_{SW}$. Choose the basis of $H^{0}\left(\mathcal{C}_{\mathbf{u}},\Omega_{\mathcal{C}_{u}}^{0}\right)$ to be the Abelian differentials of the first kind $\omega_{k}=\frac{x^{n-k}dx}{y},\,\,\, k=2,\cdots,n$. Following \cite{klemm1994monodromies}\cite{lerche1998introduction}, we have that:

\begin{lemma}
The Seiberg-Witten differential can be chosen to be 
\begin{equation}
    \lambda_{SW}=constant\cdot \left(\frac{\partial}{\partial x}\mathcal{W}_{A_{n-1}}(x,u_{i})\right)\frac{x\,dx}{y}
\end{equation}
up to an addition of exact forms.
\end{lemma}
 \begin{proof}
 Choosing $f_{k}(\mathbf{u}):=-(n-k)$, then $\lambda_{SW}=-(n-2)u_{2}\,\omega_{2}-(n-3)u_{3}\,\omega_{3}-\cdots-u_{n-1}\,\omega_{n-1}+nx^{n}$\[=\left(-(n-2)u_{2}\,x^{n-3}-(n-3)u_{3}\,x^{n-4}-\cdots-u_{n-1}+nx^{n-1}\right)\frac{x\,dx}{y}\]
 \end{proof}
 \begin{lemma}
 The Seiberg-Witten differential $\lambda_{SW}$ constructed above is residues free. Thus, when pairing with cycles, the value is invariant under continuous deformation of the cycle when crossing the poles of $\lambda_{SW}$.
 \end{lemma}
 
 \begin{proof}
 By the above lemma, $\lambda_{SW}$ is a linear combination of $\{\omega_{i}\}$, which are residues free.
 \end{proof}

\subsection{WCS in SU(2) SW integrable system}

Here $\mathcal{B}\cong\mathbb{C}$ (with coordinate $u$), $\Delta_{\Lambda}=\{\pm\Lambda^{2}\}$. The curve over $u\in\mathcal{B}^{0}$ is $\mathcal{C}_{u}: y^{2}=(x^{2}-u)^{2}-\Lambda^{4}$. The local system splits as $\underline{\Gamma}_{u}=\mathbb{Z}\alpha\oplus\mathbb{Z}\beta$, where $\{\alpha, \beta\}$ are standard basis of $\mathcal{C}_{u}$. So charge vector $\gamma$ can be written as $\gamma=q\,\alpha+g\,\beta$, where $g$ and $q$ are integers, called the ``electric " and ``magnetic" coordinates of $\gamma$. The vanishing cycles $\gamma_{\pm\Lambda^{2}}$ at $\pm\Lambda^{2}$ are given by (\cite{klemm1996nonperturbative}):$ \gamma_{-\Lambda^{2}}=\beta-2\alpha=(-2,1)$ and $\gamma_{+\Lambda^{2}}=\beta=(0,1)$ respectively. 

Let $a(u):=Z_{u}(\alpha), \, a_{D}(u):=Z_{u}(\beta)$, then $Z_{u}(\gamma)=q\,a(u)+g\,a_{D}(u)=q\oint_{\alpha}\lambda_{SW}+g\oint_{\beta}\lambda_{SW}$, where $\lambda_{SW}=\frac{2x^{2}\,dx}{y}$ (see (15)). The only wall of first kind (\cite{bilal1996curves}) in this case is given by \[\mathcal{W}^{1}=\left\{u\in\mathcal{B}: Im\left(\frac{a_{D}(u)}{a(u)}\right)=0\right\}.\]
It was shown (\cite{bilal1996curves}, or see \cite{barrett1998problem} for a rigorous mathematical treatment) that $\mathcal{W}^{1}$ (figure 1) is a closed curve passing through $\pm\Lambda^{2}$ that divides $\mathcal{B}$ into the \textit{strong coupling region} $\mathcal{B}_{strong}$, and the \textit{weak coupling region} $\mathcal{B}_{weak}$.
\begin{figure}[h]
    \centering
    \includegraphics[height=2cm, width=3.8cm]{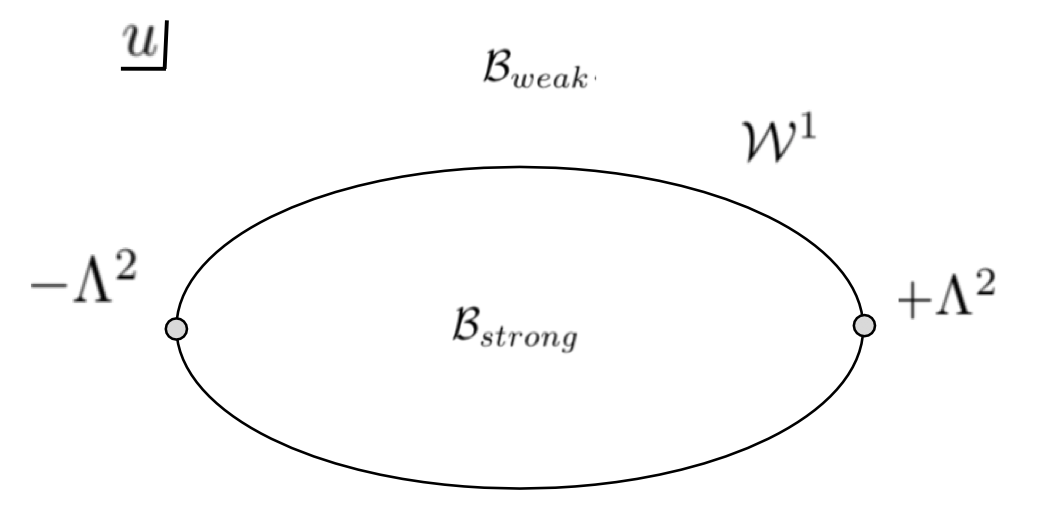}
    \caption{Curve of marginal stability in $SU(2)$ case}
\end{figure}

\textbf{Problem}: The BPS states for $SU(2)$ case are given in physics literature (c.f.,\cite{fraser1997weak}\cite{bilal1996curves}). We want to encode the charges of these states and the corresponding DT-invariants by using the WCS formalism. 
\begin{proposition}
    The attractor points in $SU(2)$ case coincide with the discriminant locus $\Delta_{\Lambda}=\{\pm\Lambda^{2}\}$.
\end{proposition}

\begin{proof}
As the ``mass" $m_{u}(\gamma):=|Z_{u}(\gamma)|$ is decreasing along the flow line $\mathcal{L}_{\gamma}$, the flow lines terminate at points where $m_{u}(\gamma)$ vanishes. But we know that only $\pm\Lambda^{2}$ support the vanishing cycles $\gamma_{\pm\Lambda^{2}}$.
\end{proof}

\begin{remark}
It is possible that the attractor flow lines terminate at a ``regular zero" $u_{m}$ of $m_{u}(\gamma)$, however, this situation is excluded by the following reason. Suppose that $Z_{u_{m}}(\gamma)$ vanishes, since $\lambda_{SW}$ does not blow up at $u_{m}$, $\gamma$ must be the vanishing cycles over $u_{m}$. But we have shown that the only possible vanishing cycles are situated at $\pm\Lambda^{2}$.
\end{remark}

\textbf{Algorithm for computing DT invariants}: The WCS is induced by $Y:\mathcal{B}\times\ S_{\theta}\to\Gamma_{\mathbb{R}}^{*}\,\,(u,\,\theta)\longmapsto Im(e^{-i\theta}Z_{u}(\gamma))$. Intuitively, we ``scan over" all possible charge vectors $\gamma$ over all points in all directions to see which charge vectors can be assumed by the BPS states. By proposition 3.5, this amounts to the (local) embedding $\mathcal{B}\hookrightarrow Stab(\mathfrak{g})$. We have the following \textbf{algorithm} for enumerate BPS states and the corresponding DT invariants:

\begin{itemize}
    \item Initial data: at the attractor points $\pm\Lambda^{2}$, we have BPS states with charges $(0,1)$ and $(2,-1)$. We assign the corresponding DT invariants to be $\Omega_{\pm\Lambda^{2}}(\gamma_{\gamma_{\pm\Lambda^{2}}})=1$. 
    \item Constructing attractor trees: Given $(u,\gamma)\in\underline{\Gamma}_{u}$, we consider all attractor trees on $\mathcal{B}$ with root at $u$, and the external vertexes ending at the attractor points. All edges of the trees are given by the attractor flows associated with $(u,\gamma)$, and each internal vertex should have valency at least three and lies in the wall of the first kind $\mathcal{W}^{1}$. Moreover, the ``\textbf{balance condition}" at each internal vertex $v$ should be satisfied, that is: If the incoming edge at $v$ is the attractor flow associated to $\gamma^{in}$, and the out coming edges are associated to charges $\gamma_{i}^{out}$, then we should have that at each internal vertex $v$: $\gamma^{in}=\sum_{i}\gamma_{i}^{out}$.
    
    \item Determining $\Omega_{u}(\gamma)$: Assuming that the number of the attractor trees is finite, we start at the attractor points, and move backward toward the root $u$. At each inner vertex, apply the KSWCF, so that $\Omega_{v}(\gamma^{in})$ can be computed from $\Omega_{v}(\gamma^{out})$. We will get $\Omega_{u}(\gamma)$ inductively.
\end{itemize}

\textbf{WCS in the strong coupling region}: Given $(u_{0},\gamma)\in\underline{\Gamma}_{u_{0}}$, where $u_{0}\in\mathcal{B}_{strong}$. In order for the line $\mathcal{L}_{u_{0},\gamma}$ to end at the attractor points, the only possible choices of $\gamma$ are the vanishing cycles $\gamma_{\pm\Lambda^{2}}$ at the attractor points $\pm\Lambda^{2}$. It was proved in \cite{fayyazuddin1997results} that $\mathcal{L}_{u_{0},\gamma}$ is the unique line that connect $u_{0}$ to either $\pm\Lambda^{2}$ depending on if $\gamma=\gamma_{\pm\Lambda^{2}}$ respectively. 

We thus conclude that in the strong coupling region $\mathcal{B}_{strong}$, the only attractor flow issuing from a point that terminates at attractor points is the straight line (in affine structure) that connects it to the corresponding attractor point $\pm\Lambda^{2}$. Consequently, within $\mathcal{B}_{strong}$, the only BPS charges are those corresponding to vanishing cycles $\gamma_{\pm\Lambda^{2}}$ at $\pm\Lambda^{2}$, namely, the magnetic monopole with charge vector $(0,1)$ and the dyon with charge vector $(2,-1)$.

\textbf {WCS in the weak coupling region}: By the algorithm described before, we need to consider all rooted attractor trees with root at $u_{0}\in\mathcal{B}_{weak}$ that terminates at the attractor points $\pm\Lambda^{2}$. We expect the following scenario:

\begin{itemize}
    \item For $\gamma=\gamma_{\pm\Lambda^{2}}$, the flow line $\mathcal{L}_{u_{0},\gamma}$ will directly flow into the attractor points $\{\pm\Lambda^{2}\}$ without splitting, which means the two states with charges $\gamma_{\pm\Lambda^{2}}$ persists in the weak coupling region.
    \item For other charge $\gamma$, the flow line $\mathcal{L}_{u_{0},\gamma}$  will first intersect $\mathcal{W}^{1}$ at the point $u_{*}$. At this point, the flow splits into two flows $\mathcal{L}_{u_{*},m\gamma_{+\Lambda^{2}}}$ and $\mathcal{L}_{u_{*},n\gamma_{-\Lambda^{2}}}$ that start at $u_{*}$ and terminate at the attractor points $\pm\Lambda^{2}$ respectively. i.e., $\mathcal{L}_{u_{0},\gamma}=\mathcal{L}_{u_{*},m\gamma_{+\Lambda^{2}}}+\mathcal{L}_{u_{*},n\gamma_{-\Lambda^{2}}}$ corresponding to $\gamma=m\gamma_{+\Lambda^{2}}+n\gamma_{-\Lambda^{2}}$. 
    If such a \textit{split attractor flow} associate to $\gamma$ exists with $\Omega_{u_{0}}(\gamma)\neq 0$, then the corresponding BPS state with charge $\gamma$ must exist in $\mathcal{B}_{weak}$, and its DT-invariant $\Omega_{u_{0}}(\gamma)$ can be computed by applying the KSWCF at the split point $u_{*}$.
\end{itemize}

\begin{lemma}
The flow lines $\mathcal{L}_{u_{*},m\gamma_{+\Lambda^{2}}}$ and $\mathcal{L}_{u_{*},n\gamma_{-\Lambda^{2}}}$ coming from $\pm\Lambda^{2}$ can only intersect at the wall $\mathcal{W}^{1}$.
\end{lemma}

\begin{proof} Since $\mathcal{L}_{u_{*},m\gamma_{+\Lambda^{2}}}=\left\{u\in\mathcal{B}:Arg\,Z_{u}(\gamma_{+\Lambda^{2}})=\theta_{0}\right\}$ and $\mathcal{L}_{u_{*},n\gamma_{-\Lambda^{2}}}=\left\{u\in\mathcal{B}:Arg\,Z_{u}(\gamma_{-\Lambda^{2}})=\theta_{0}\right\}$, we see that $\mathcal{L}_{u_{*},m\gamma_{+\Lambda^{2}}}\cap\mathcal{L}_{u_{*},n\gamma_{-\Lambda^{2}}}=\left\{u^{*}\in\mathcal{B}:Arg\,Z_{u^{*}}(\gamma_{-\Lambda^{2}})=Arg\,Z_{u^{*}}(\gamma_{+\Lambda^{2}})\right\}$, which belongs to $\mathcal{W}^{1}$.
\end{proof}  

\begin{lemma}
The intersection point $u^{*}$ coincides with the splitting point $u_{*}$.
\end{lemma}

\begin{proof}
Suppose that $u^*\ne u_*$, then the phase of $Z_{u_{*}}(\gamma)$ would be different from that of $Z_{u^{*}}(\gamma_{\pm\Lambda^{2}})$, contradicting to the fact that $Arg\,Z_{u_{*}}(\gamma)= Arg\,Z_{u^{*}}(\gamma_{\pm\Lambda^{2}})=\theta_{0}$.
\end{proof}

Now we ask the question: When will a splitting attractor flow tree like above exist? Putting it in another way: given a point $u_{0}\in\mathcal{B}_{weak}$, for which charges $\gamma\in\underline{\Gamma}_{u_{0}}$, does the attractor flow tree like above exist?

To answer this question, we need to reverse the logic and consider instead the two attractor flow lines $\mathcal{L}_{+\Lambda^{2}}$ and  $\mathcal{L}_{-\Lambda^{2}}$ coming from the attractor points $\pm\Lambda^{2}$, which correspond to the same $\theta$-value. As the lemmas above show, these two ``incoming rays" $\mathcal{L}_{\pm\Lambda^{2}}$ will intersect exactly at one point $u_{*}\in\mathcal{W}^{1}$. And by KSWCF, the two rays ``scatter" at $u_{*}$ into possibly infinity many rays $\mathcal{L}_{m,n}$ with all possible charges $m\gamma_{+\Lambda^{2}}+n\gamma_{-\Lambda^{2}}$. Together with the incoming rays, they form the so-called \textit{scattering diagram}.

Then by considering a small loop around the scattering point $u_{*}$, KSWCF is equivalent (informally) to the triviality of the monodromy, i.e., $\stackrel{\circlearrowright}{\underset{t_{i}}{{\prod}}}\,\mathbb{S}(\textit{l}_{\gamma_{i}})=\,id$, where the product is taken in increasing order of elements $t_{i}$. Written in another way, this read as: $\prod_{m,n\geq0;\,m/n\nearrow}\mathcal{K}_{(m,n)}^{\Omega_{u^{*}_{-}}(m,n)}=\prod_{m,n\geq0;\,m/n\searrow}\mathcal{K}_{(m,n)}^{\Omega_{u^{*}_{+}}(m,n)}$  (see (8)). In our case, $k=\langle\gamma_{-\Lambda^{2}},\gamma_{+\Lambda^{2}}\rangle=2$, the above specializes into (formula (10)):
\begin{equation}
    \mathcal{K}_{2,-1}\,\mathcal{K}_{0,1}=(\mathcal{K}_{0,1}\,\mathcal{K}_{2,1}\,\mathcal{K}_{4,1}\,\cdot\cdot\cdot)\,\mathcal{K}_{2,0}^{-2}\,(\cdot\cdot\cdot\,\mathcal{K}_{6,-1}\,\mathcal{K}_{4,-1}\,\mathcal{K}_{2,-1})
\end{equation}

We see that after ``scattering", we get countably many BPS states with the charge vectors indicated by the sub-index on the RHS of the above formula, while the corresponding BPS invariants are indicated by the super-index on the RHS of the formula above.

Using WCS formalism, we thus produce the following \textbf{algorithm} that  ``foresees" these BPS states. Given $(u_{0},\gamma)\in\underline{\Gamma}_{u_{0}}$, where $u_{0}\in\mathcal{B}_{weak}$, to test if the BPS state with charge vector $\gamma$ exists or not (and if it exists, determine its BPS invariant $\Omega_{u_{0}}(\gamma)$) at $u_{0}$, we do the followings

\begin{itemize}
    \item First, determine $\theta_{0}=arg\,Z_{u_{0}}(\gamma)$, and consider the gradient flow associated to $F_{\gamma}(u)=Re(e^{-i\theta_{0}}Z_{\gamma}(u))$.
    \item Second, the rooted attractor flow tree with root $u_{0}$ would be generically a tree that splits only at some point $u_{*}\in\mathcal{W}^{1}$, with the two external edges ending at the two attractor points $\{\pm\Lambda^{2}\}$.
    \item Third, start with the attractor points, tracing along the external edges till the scattering point $u_{*}$, where we use KSWCF, from which we can tell if the charge $\gamma$ coincides one of the charges on the RHS of (16). If the answer is YES, then the state with charge vector $\gamma$ exists at $u_{0}$, and the DT invariant $\Omega_{u_{0}}(\gamma)$ can be read off from (16); If NO, then at $u_{0}$ there does not exist BPS state with charge $\gamma$.  
\end{itemize}

\subsection{WCS in SU(3) SW integrable system}

 In this case, $\mathcal{B}=\mathbb{C}^{2}$ with coordinate $\mathbf{u}:=(u,v)$, and $\mathcal{C}_{\mathbf{u}}:\, y^{2}=(x^{3}-ux-v)^{2}-\Lambda^{6}$ with $\lambda_{SW}=\frac{(3x^{2}-u)\,x\,dx}{y}$. The discriminant locus is given by $\Delta_{\Lambda}:=\Delta_{\Lambda}^{+}\cup\Delta_{\Lambda}^{-}=\left\{(u,v)\in\mathbb{C}^{2}:\,4u^{3}=27(v\pm\Lambda^{3})^{2}\right\}$. With a choice of symplectic basis $\{\alpha^{1},\alpha^{2};\beta_{1},\beta_{2}\}$ of $\underline{\Gamma}_{\mathbf{u}}$, a charge can be written as $\gamma=(\mathbf{g}\,\,\mathbf{q})=g^{1}\beta_{1}+g^{2}\beta_{2}+q_{1}\alpha^{1}+q_{2}\alpha^{2}$, where $\mathbf{g}=\left(g^{1}\,\,g^{2}\right)$ denotes the ``magnetic charge" of $\gamma$ while $\mathbf{q}=\left(q_{1}\,\,q_{2}\right)$ the ``electric charge" of $\gamma$. The central charge is given by $Z_{\mathbf{u}}(\gamma)=\oint_{\gamma}\,\lambda_{SW}= \mathbf{g}\cdot\mathbf{a}_{D}+\mathbf{q}\cdot\mathbf{a}$, where $\mathbf{a}=\left(\oint_{\alpha^{1}}\,\lambda_{SW}\,\,\oint_{\alpha^{1}}\,\lambda_{SW}\right)^{T}$ and $\mathbf{a}_{D}:=\left(\oint_{\beta_{1}}\,\lambda_{SW}\,\,\oint_{\beta_{2}}\,\lambda_{SW}\right)^{T}$.

\noindent\textbf{Strong coupling spectrum}: The \textit{strong coupling region} in $SU(3)$ case is characterized by small values of $u$ and $v$ (c.f.,\cite{galakhov2013wild}) . It is known that in this region (\cite{klemm1994monodromies}\cite{klemm1996nonperturbative}\cite{klemm1995simple}) there exists only six BPS states with vanishing ``mass" function. These states correspond to the six vanishing cycles $\nu_{i}^{\pm}$,  $i=1,2,3$ that are the invariant cycles under the Picard-Lefschetz monodromies around the six singular lines $\Sigma_{\pm}^{i}$, $i=1,2,3$ below. This can be seen by cutting the base by hyperplane $H=\{Imv=0\}$, in which $\Delta_{\Lambda}$ becomes three pairs of lines $ \mathscr{L}_{k}=\Sigma^{k}_{+}\cup \Sigma_{-}^{k},\,\, k=1,2,3$ (see \cite{klemm1996nonperturbative}, from which the figure 2 below is taken). There exists a symplectic basis  such that in this basis (\cite{hollowood1997strong}): 
\begin{equation*}
    \nu_{1}^{+}=(1,0;-2,1)\,\nu_{1}^{-}=(1,0;0,0);\,\,
    \nu_{2}^{+}=(0,1;0,0)\,\nu_{2}^{-}=(0,1;-1,2);\,\,
         \nu_{3}^{+}=(1,1; -2,1)\,\nu_{3}^{-}=(1,1; -1,2)
\end{equation*}

 \begin{figure}[h]
    \centering
\includegraphics[height=4.5cm, width=5cm]{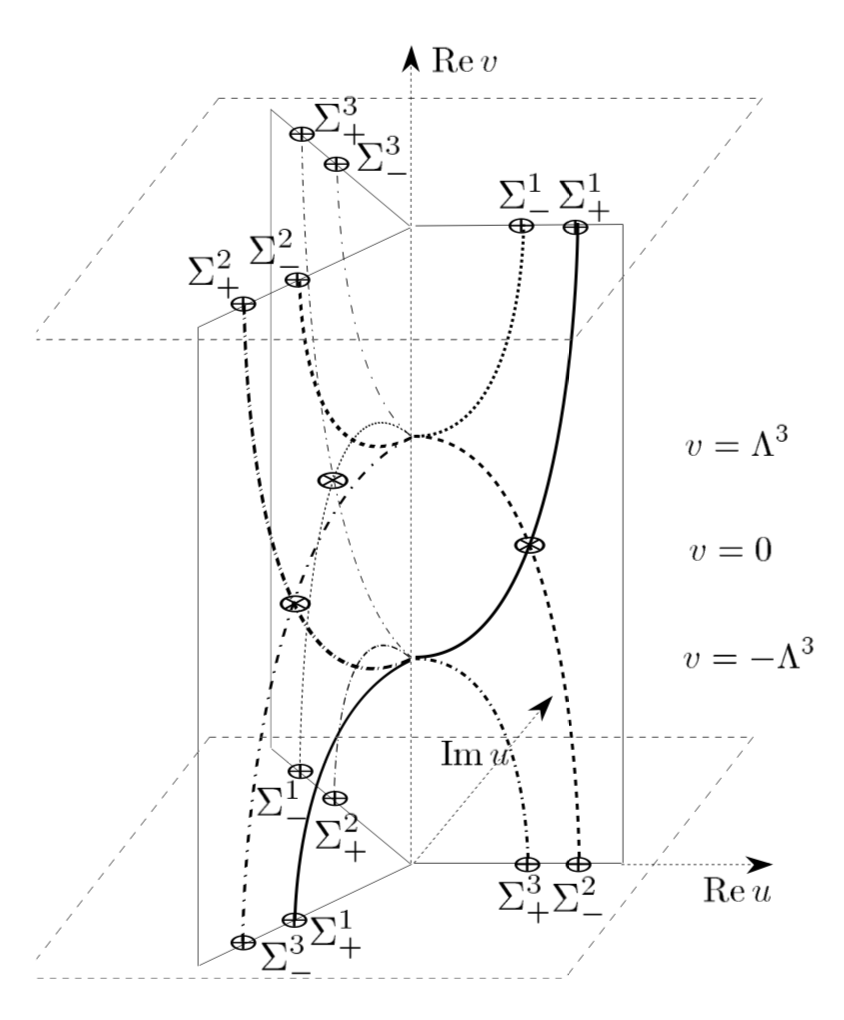}
    \caption{Geometry of the base for $SU(3)$ when $Im\, v=0$, figure taken from \cite{klemm1996nonperturbative}.}
\end{figure}

Above vanishing cycles can be organized by the roots and co-roots of $SU(3)$ (\cite{taylor2001strong}\cite{klemm1995simple}). "Electric charges" are expressed in terms of the simple roots of $SU(3)$, i,e.,  $\alpha_{1}=(2,-1)$ and $\alpha_{2}=(-1,2)$. The ``magnetic charge vectors" sit in the weight lattice $\Lambda_{w}$. They are expressed in terms of the co-roots, i.e., vector $(n_{1},n_{2})$ of the form $\mathbf{g}=n_{1}\alpha_{1}^{\vee}+n_{2}\alpha_{2}^{\vee}\in\Lambda_{w}$. In our case, $\alpha^{\vee}=\alpha$. So we express the above vanishing cycles as:
\[\nu_{1}^{+}=(\alpha_{1},-\alpha_{1})\quad\nu_{1}^{-}=(\alpha_{1},0);\]
\begin{equation}
        \nu_{2}^{+}=(\alpha_{2},0)\quad\nu_{2}^{-}=(\alpha_{2},\alpha_{2});
\end{equation}
\[\nu_{3}^{+}=(\alpha_{1}+\alpha_{2},-\alpha_{1})\quad\nu_{3}^{-}=(\alpha_{1}+\alpha_{2},\alpha_{2})\]

\begin{remark}
Since $\alpha_{3}=\alpha_{1}+\alpha_{2}$, so $\nu_{3}^{\pm}=\nu_{1}^{\pm}+\nu_{2}^{\pm}$. The fact that each $\alpha_{k}$ corresponds to an embedding of $su(2)$ into $su(3)$ is reflected by the fact that each pair $\Sigma_{\pm}^{k}$ corresponds to a WCS similar to $SU(2)$ case considered before.
\end{remark}

\begin{figure}[h]
    \centering
    \includegraphics[height=2cm, width=4.5cm]{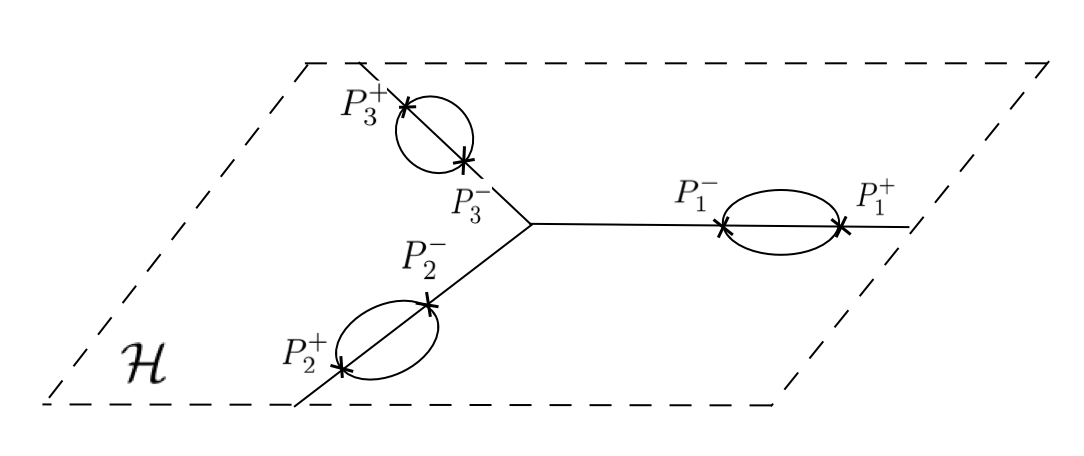}
    \caption{walls restricted to the plane $\mathcal{H}$, each corresponding to a situation similar to $SU(2)$ case.}
\end{figure}

The \textbf{weakly coupled region} is the region on $\mathcal{B}$ where $|\alpha\cdot\mathbf{a}|\gg\Lambda$ for all positive roots 
$\alpha$ (c.f.,\cite{fraser1997weak}). 

\noindent\textbf{Problem}: \textit{Starting with BPS spectrum (17) (\textbf{initial data}), compute the weak coupling BPS spectrum using WCS.} 

To study WCS in weak coupling region, cut the base further by the plane $H_{\text{cut}}:=\{z=Re (v)\equiv constant>>\Lambda^{3}\}$, which intersects with the six singular lines at points $P_{k}^{\pm}:=\Sigma^{i}_{\pm}\cap H_{\text{cut}}$. We investigate the WCS on the resulting plane $\mathcal{H}:=\mathcal{B}\cap H\cap H_{\text{cut}}$.

\noindent\textbf{Walls of the first kind} (similar to $SU(2)$ wall): We expect certain BPS states decay into $\nu_{i}^{\pm}$s. As each pair $\nu_{i}^{\pm}$ corresponds to a $SU(2)$-situation, we have three walls of the first kind near $\Delta_{\Lambda}$ similar to the wall in $SU(2)$ case: $\mathcal{W}_{k}^{1}=\mathcal{W}^{1}(\nu_{k}^{+},\nu_{k}^{-}):=\left\{\mathbf{u}=(u,v)\in\mathbb{C}^{2}:\, Arg(Z_{\mathbf{u}}(\nu_{k}^{+})=Arg(Z_{\mathbf{u}}(\nu_{k}^{-}))\right\},\, k=1,2,3$. It was pointed out in \cite{argyres1995new} and \cite{giveon1995effective} that these walls have the topology type $S^{1}\times \mathbb{C}$ for large $\mathbf{u}$. When restricted to $\{Im(v)=0\}$, these walls have the cylinder topology, i.e., $S^{1}\times\mathbb{R}$. Numerical test about their shapes are given in \cite{campbellova2010supersymetric}, from which the following figure 4 is copied below. 

\begin{figure}[h]
    \centering
    \includegraphics[height=3.0cm, width=4.4cm]{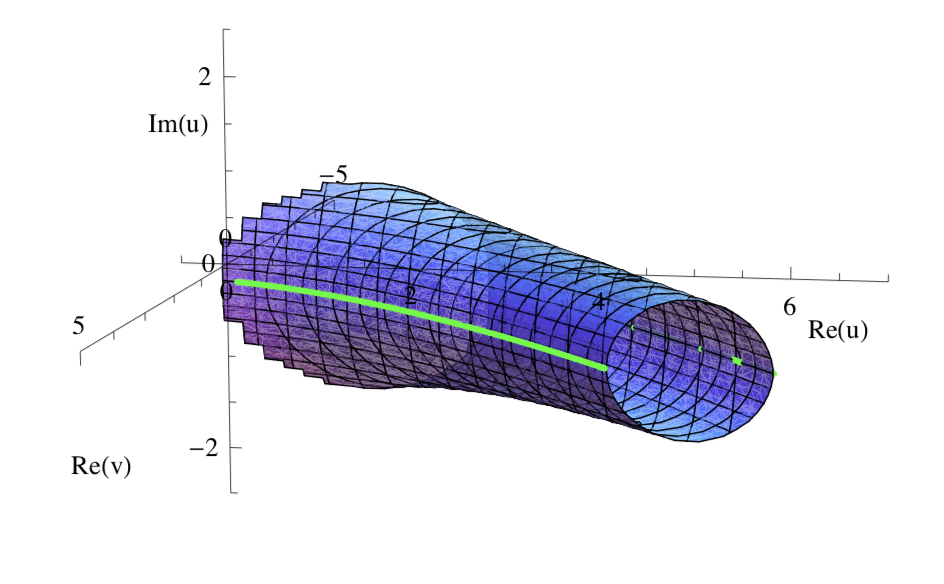}
    \caption{$\mathcal{W}^{1}_{2}$ away from singular locus when $Im(v)=0$, $Re(v)>0$, figure taken from \cite{campbellova2010supersymetric}}
\end{figure}
 
Numerical test in section 6.4.1 of \cite{campbellova2010supersymetric} shows that cutting $\mathcal{B}$ further by $H_{\text{cut}}$ in figure 4 corresponds to taking a slice of cylinder (wall of the first kind in $Im(v)=0$ hyperplane). Notice that the green lines in figure 4 are the singular lines. We will see that the the two intersection points $P_{k}^{\pm}$ play the same role as $\pm\Lambda^{2}$ played in $SU(2)$ case.
\begin{proposition}
The walls of the first kind $\mathcal{W}^{1}_{k}$, when restricted to the plane $\mathcal{H}$, is topologically homeomorphic to the circle $S^{1}$, and plays the same role as the wall of stability in $SU(2)$ case.
\end{proposition}

\begin{proof} (sketch)
The variable $v$ is ``frozen", and $\mathcal{H}$ is endowed with coordinate $u$. We know from \cite{klemm1996nonperturbative} and \cite{campbellova2010supersymetric} that the period integrals $\mathbf{a}$ and $\mathbf{a_{D}}$ can be expressed in terms of Appell function \[F_{4}(a,b,c,c^{\prime},u,v)=\sum_{m=0}^{\infty}\sum_{n=0}^{\infty}\frac{(a)_{m+n}(b)_{m+n}}{(c)_{n}(c^{\prime})_{m}(1)_{n}(1)_{m}}u^{m}v^{n}\]
where $(a)_{n}$ is the Pochhammer symbol $(a)_{n}:=\frac{\Gamma(a+n)}{\Gamma(a)}$. However, if one of the variables $u$, $v$ is set to be constant, the Appell function reduces to the hypergeometric function used in expressing the period integrals in $SU(2)$ case. Thus, the situation will be reduced to the $SU(2)$ case, to which the results about the shape of the $SU(2)$-wall applies. 
\end{proof}

\begin{proposition}
The WCS formalism, when applied to $\mathcal{H}$, produces the BPS states with charges given by:
\[n\nu_{k}^{+}+(n-1)\nu_{k}^{-},\quad k=1,2,3;\,\,n=1,2,\cdots\]
\begin{equation}
    \nu_{k}^{+}+\nu_{k}^{-},\quad k=1,2,3;
\end{equation}
\[(n-1)\nu_{k}^{+}+n\nu_{k}^{-},\quad k=1,2,3;\,\,n=1,2,\cdots\]

All states have DT-invariants $\Omega=1$, except for the middle row states, which have DT-invariants $\Omega=-2$.
\end{proposition}

\begin{proof}

As $\langle\nu_{k}^{+},\nu_{k}^{-}\rangle=2$, for $k=1,2,3$, we apply the KSWCF (10) to each pair $\nu_{k}^{\pm}$. 
\end{proof}

In terms of the positive roots $\{\alpha_{1},\alpha_{2},\alpha_{3}:=\alpha_{1}+\alpha_{2}\}$, the charge vectors in (23) can be expressed as: for $k=1,2$, we have BPS states with charges
\begin{equation}
\{(\alpha_{1},n\alpha_{1}),\,(0,\alpha_{1}),\,(-\alpha_{1},(n-1)\alpha_{1})\}
\end{equation}
\begin{equation}
    \{(\alpha_{2},(n-1)\alpha_{2}),\,(0,\alpha_{2}),\,(-\alpha_{2},n\alpha_{2})\}
\end{equation}

While for $k=3$, we have the following BPS states with charges
\begin{equation}
    \{(\alpha_{3},\alpha_{1}+(n-1)\alpha_{3}),\,(0,\alpha_{3}),\,(-\alpha_{3},\alpha_{2}+(n-1)\alpha_{3})\}
\end{equation}

$SU(3)$ case is interesting in that when $\Lambda\to0$, there exists wall $\mathcal{W}_{\alpha_{1},\alpha_{2}}^{1}$ responsible for states in (21) to ``decay" into states in (19) and (20). For example, $(\alpha_{3},\alpha_{1}+(n-1)\alpha_{3})=(\alpha_{1},n\alpha_{1})+(\alpha_{2},(n-1)\alpha_{2})$ means that the state with charge $(\alpha_{3},\alpha_{1}+(n-1)\alpha_{3})$, besides being created by "scattering" $\nu_{3}^{+}$ and $\nu_{3}^{-}$, can also be created on one side of $\mathcal{W}_{\alpha_{1},\alpha_{2}}^{1}$ by "scattering" $(\alpha_{1},n\alpha_{1})$ and $(\alpha_{2},(n-1)\alpha_{2})$ (the relevant KSWCF is given by pentagon identity (9)).

\begin{proposition}
(\cite{taylor2001strong}\cite{hollowood1997strong}\cite{kuchiev2008charges}) The wall (of the first kind) at weak coupling region is given by
\[\mathcal{W}^{1}_{\alpha_{1},\alpha_{2}}=\left\{
    (u,v)\in\mathbb{C}^{2}:\,Im\left(\frac{\alpha_{1}\cdot\mathbf{a}}{\alpha_{2}\cdot\mathbf{a}}\right)=Im\left(\frac{a^{1}}{a^{2}}\right)\equiv 0\right\}\]
\end{proposition}

\begin{proof}
It is known (\cite{hollowood1997strong}) that in the weak coupling region, we have the following expression \[\mathbf{a}_{D}(\mathbf{a})=\frac{i}{2\pi}\sum_{\alpha\,\text{positive roots}}\alpha (\alpha\cdot\mathbf{a})\left[\ln\left(\frac{\alpha\cdot\mathbf{a}}{\Lambda}\right)^{2}\right]\]

Then, as $\Lambda\to 0$ (or equivalently, $|\alpha\cdot\mathbf{a}|\gg0$ for all positive roots 
$\alpha$), \,$\mathbf{a}_{D}\sim (C\cdot\ln\Lambda)\,\mathbf{a}$ for some constant $C$.Thus, in this limit, $|\mathbf{a}_{D}|\gg|\mathbf{a}|$, which implies that $Z(\gamma)=\mathbf{a}_{D}\cdot\mathbf{g}+\mathbf{a}\cdot\mathbf{q}\approx (\ln\Lambda)\,\mathbf{a}\cdot\mathbf{g}$. This means that in this region, the central charge is dominated by the magnetic charge, thus follows the proposition.
\end{proof}

\begin{figure}[h]
    \centering
    \includegraphics[height=3.5cm, width=5.5cm]{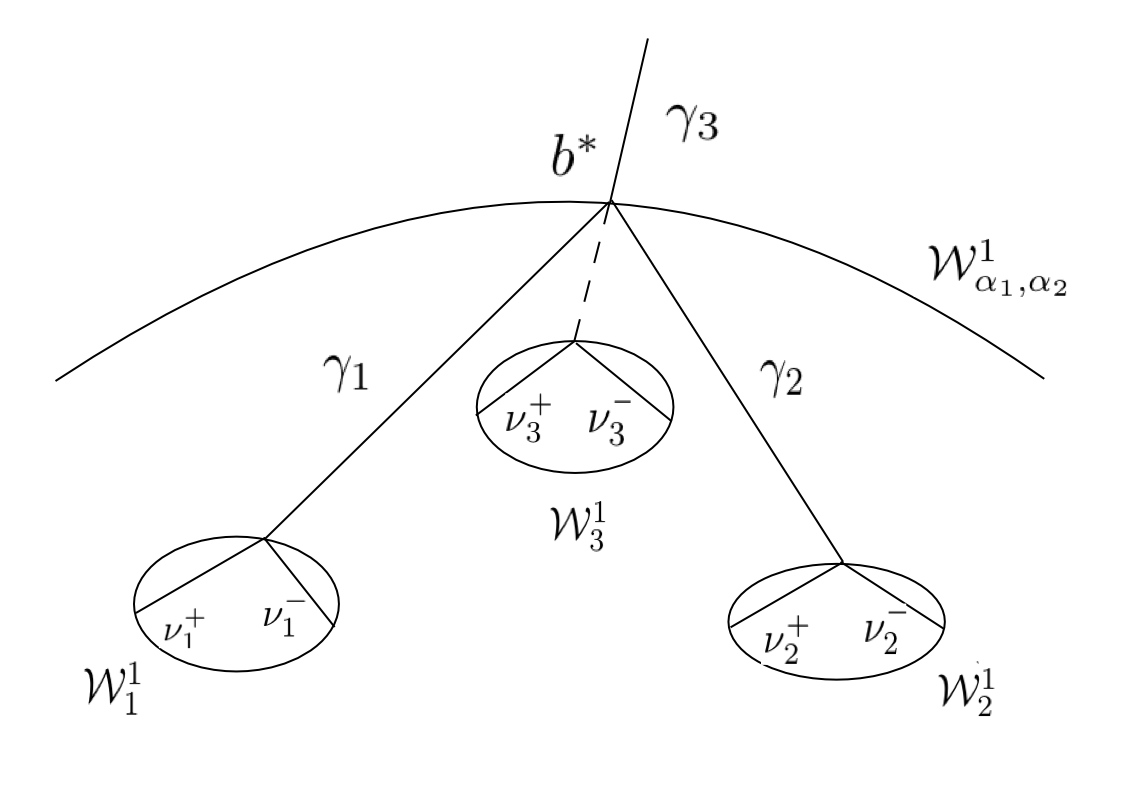}
    \caption{Attractor trees for $\gamma_{3}$ outside $\mathcal{W}_{\alpha_{1},\alpha_{2}}$}
    \label{fig:my_label}
\end{figure}

Therefore, on one side of $\mathcal{W}_{\alpha_{1},\alpha_{2}}^{1}$ where the walls $\mathcal{W}_{k}^{1}$'s are present, the states with charge  $(\alpha_{3},\alpha_{1}+(n-1)\alpha_{3})$ exist with DT-invariant equals to $1$ (following from $k=2$ KSWCF), while on the other side, the same states derive their existence in two ways. Figure 5 above illustrates the corresponding split attractor flow. Thus, we expect the DT-invariant $\Omega((\alpha_{3},\alpha_{1}+(n-1)\alpha_{3}))$  to "jump" from $1$ to $2$ after crossing the wall $\mathcal{W}_{\alpha_{1},\alpha_{2}}^{1}$. The KSWCF to be used at the split point $b^{*}$ is given by the proposition 4.11 below. 

Using short notation: $\gamma_{1}:=(\alpha_{1},n\alpha_{1})$, $\gamma_{2}:=(\alpha_{2},(n-1)\alpha_{2})$ and $\gamma_{3}:=(\alpha_{3},\alpha_{1}+(n-1)\alpha_{3})$, we have that

\begin{proposition}
 The KSWCF at the split point $b^{*}$ is given by:
 
 \begin{equation}
     \mathcal{K}_{\gamma_{1}}\mathcal{K}_{\gamma_{3}}\mathcal{K}_{\gamma_{2}}=\Xi^{+}\,\mathcal{K}_{\gamma_{3}}^{2}
\mathcal{K}_{2\gamma_{3}}^{-2}\,\Xi^{-}\,\mathcal{K}_{\gamma_{1}}
 \end{equation}
where
\begin{equation*}
    \Xi^{+}:=\prod_{n=1}^{\infty}\mathcal{K}_{n\gamma_{2}+(n-1)(\gamma_{1}+\gamma_{3})} \,\mathcal{K}_{n\gamma_{2}+(n-1)(\gamma_{1}+\gamma_{3})+\gamma_{3}}
\end{equation*}
\begin{equation*}
    \Xi^{-}:=\prod_{n=\infty}^{1}\mathcal{K}_{(n-1)\gamma_{2}+n(\gamma_{1}+\gamma_{3})+\gamma_{3}} \,\mathcal{K}_{(n-1)\gamma_{2}+n(\gamma_{1}+\gamma_{3})}
\end{equation*}
\end{proposition}

\begin{proof}
Since $\langle\gamma_{1},\gamma_{2}\rangle=\langle\gamma_{1},\gamma_{3}\rangle=1$, by applying the pentagon identity (9), we have the following
\begin{equation}
\mathcal{K}_{\gamma_{1}}\mathcal{K}_{\gamma_{3}}\mathcal{K}_{\gamma_{2}}=\mathcal{K}_{\gamma_{3}}\mathcal{K}_{\gamma_{1}+\gamma_{3}}\mathcal{K}_{\gamma_{1}}\mathcal{K}_{\gamma_{2}}=\mathcal{K}_{\gamma_{3}}\mathcal{K}_{\gamma_{1}+\gamma_{3}}\mathcal{K}_{\gamma_{2}}\mathcal{K}_{\gamma_{1}+\gamma_{2}}\mathcal{K}_{\gamma_{1}}    
\end{equation}

Note that $\langle\gamma_{1}+\gamma_{3},\gamma_{2}\rangle=\langle\gamma_{1},\gamma_{2}\rangle+\langle\gamma_{3},\gamma_{2}\rangle=2$, so by applying (10), the right hand side of (23) becomes  
\[\mathcal{K}_{\gamma_{3}}\prod_{n=1}^{\infty}\mathcal{K}_{n\gamma_{2}+(n-1)(\gamma_{1}+\gamma_{3})}\,\mathcal{K}_{\gamma_{1}+\gamma_{2}+\gamma_{3}}^{-2}\prod_{n=\infty}^{1}\mathcal{K}_{(n-1)\gamma_{2}+n(\gamma_{1}+\gamma_{3})}\,\mathcal{K}_{\gamma_{3}}\,\mathcal{K}_{\gamma_{1}}=P\,\mathcal{K}_{2\gamma_{3}}^{-2}\,Q\,\mathcal{K}_{\gamma_{1}}\]

where \[P:=\mathcal{K}_{\gamma_{3}}\prod_{n=1}^{\infty}\mathcal{K}_{n\gamma_{2}+(n-1)(\gamma_{1}+\gamma_{3})}\] and \[Q:=\prod_{n=\infty}^{1}\mathcal{K}_{(n-1)\gamma_{2}+n(\gamma_{1}+\gamma_{3})}\,\mathcal{K}_{\gamma_{3}}\]

We now simplify $P$ and $Q$. By repeated use of pentagon identities, we get that
\[P=\mathcal{K}_{\gamma_{3}}\left(\mathcal{K}_{\gamma_{2}}\,\mathcal{K}_{2\gamma_{2}+(\gamma_{1}+\gamma_{3})}\,\mathcal{K}_{3\gamma_{2}+2(\gamma_{1}+\gamma_{3})}\cdots\right)\]
\[=\mathcal{K}_{\gamma_{2}}\,\mathcal{K}_{\gamma_{2}+\gamma_{3}}\,\mathcal{K}_{\gamma_{3}}\,\mathcal{K}_{2\gamma_{2}+(\gamma_{1}+\gamma_{3})}\,\mathcal{K}_{3\gamma_{2}+2(\gamma_{1}+\gamma_{3})}\,\mathcal{K}_{4\gamma_{2}+3(\gamma_{1}+\gamma_{3})}\cdots\]
\[=\mathcal{K}_{\gamma_{2}}\,\mathcal{K}_{\gamma_{2}+\gamma_{3}}\,\mathcal{K}_{2\gamma_{2}+(\gamma_{1}+\gamma_{3})}\,\mathcal{K}_{2\gamma_{2}+(\gamma_{1}+2\gamma_{3})}\,\mathcal{K}_{\gamma_{3}}\,\mathcal{K}_{3\gamma_{2}+2(\gamma_{1}+\gamma_{3})}\,\mathcal{K}_{4\gamma_{2}+3(\gamma_{1}+\gamma_{3})}\cdots\]
\[=\mathcal{K}_{\gamma_{2}}\,\mathcal{K}_{\gamma_{2}+\gamma_{3}}\,\mathcal{K}_{2\gamma_{2}+(\gamma_{1}+\gamma_{3})}\,\mathcal{K}_{2\gamma_{2}+(\gamma_{1}+2\gamma_{3})}\,\mathcal{K}_{3\gamma_{2}+2(\gamma_{1}+\gamma_{3})}\,\mathcal{K}_{3\gamma_{2}+2(\gamma_{1}+\gamma_{3})+\gamma_{3}}\,\mathcal{K}_{\gamma_{3}}\,\mathcal{K}_{4\gamma_{2}+3(\gamma_{1}+\gamma_{3})}\cdots\]\[=\left(\mathcal{K}_{\gamma_{2}}\,\mathcal{K}_{\gamma_{2}+\gamma_{3}}\right)\left(\mathcal{K}_{2\gamma_{2}+(\gamma_{1}+\gamma_{3})}\,\mathcal{K}_{2\gamma_{2}+(\gamma_{1}+\gamma_{3})+\gamma_{3}}\right)\left(\mathcal{K}_{3\gamma_{2}+2(\gamma_{1}+\gamma_{3})}\,\mathcal{K}_{3\gamma_{2}+2(\gamma_{1}+\gamma_{3})+\gamma_{3}}\right)\,\mathcal{K}_{\gamma_{3}}\,\mathcal{K}_{4\gamma_{2}+3(\gamma_{1}+\gamma_{3})}\cdots\]
\begin{equation}
    =\cdots=\prod_{n=1}^{\infty}\mathcal{K}_{n\gamma_{2}+(n-1)(\gamma_{1}+\gamma_{3})}\,\mathcal{K}_{n\gamma_{2}+(n-1)(\gamma_{1}+\gamma_{3})+\gamma_{3}}\,\mathcal{K}_{\gamma_{3}}=\Xi^{+}\,\mathcal{K}_{\gamma_{3}}
\end{equation}

Similarly, we have that
\[Q=\left(\cdots\mathcal{K}_{3\gamma_{2}+4(\gamma_{1}+\gamma_{3})}\,\mathcal{K}_{2\gamma_{2}+3(\gamma_{1}+\gamma_{3})}\,\mathcal{K}_{\gamma_{2}+2(\gamma_{1}+\gamma_{3})}\,\mathcal{K}_{\gamma_{1}+\gamma_{3}}\right)\,\mathcal{K}_{\gamma_{3}}\]
\[=\cdots\mathcal{K}_{3\gamma_{2}+4(\gamma_{1}+\gamma_{3})}\,\mathcal{K}_{2\gamma_{2}+3(\gamma_{1}+\gamma_{3})}\,\mathcal{K}_{\gamma_{2}+2(\gamma_{1}+\gamma_{3})}\,\mathcal{K}_{\gamma_{3}}\,\mathcal{K}_{\gamma_{1}+2\gamma_{3}}\,\mathcal{K}_{\gamma_{1}+\gamma_{3}}\]
\[\cdots\mathcal{K}_{3\gamma_{2}+4(\gamma_{1}+\gamma_{3})}\,\mathcal{K}_{2\gamma_{2}+3(\gamma_{1}+\gamma_{3})}\,\mathcal{K}_{\gamma_{3}}\,\mathcal{K}_{\gamma_{2}+2(\gamma_{1}+\gamma_{3})+\gamma_{3}}\,\mathcal{K}_{\gamma_{2}+2(\gamma_{1}+\gamma_{3})}\]
\[\cdots\mathcal{K}_{3\gamma_{2}+4(\gamma_{1}+\gamma_{3})}\,\mathcal{K}_{\gamma_{3}}\,\mathcal{K}_{2\gamma_{2}+3(\gamma_{1}+\gamma_{3})+\gamma_{3}}\,\mathcal{K}_{2\gamma_{2}+3(\gamma_{1}+\gamma_{3})}\,\mathcal{K}_{\gamma_{2}+2(\gamma_{1}+\gamma_{3})}\]
\begin{equation}
    =\cdots=\mathcal{K}_{\gamma_{3}}\,\prod_{n=\infty}^{1}\mathcal{K}_{(n-1)\gamma_{2}+n(\gamma_{1}+\gamma_{3})+\gamma_{3}}\,\mathcal{K}_{(n-1)\gamma_{2}+n(\gamma_{1}+\gamma_{3})}=\mathcal{K}_{\gamma_{3}}\,\Xi^{-}
\end{equation}

Putting (23), (24) and (25) together, we get finally that 
\[\mathcal{K}_{\gamma_{1}}\,\mathcal{K}_{\gamma_{3}}\,\mathcal{K}_{\gamma_{2}}=P\,\mathcal{K}_{2\gamma_{3}}^{-2}\,Q\,\mathcal{K}_{\gamma_{1}}=\Xi^{+}\,\mathcal{K}_{\gamma_{3}}\,\mathcal{K}_{2\gamma_{3}}^{-2}\,\mathcal{K}_{\gamma_{3}}\,\Xi^{-}\,\mathcal{K}_{\gamma_{1}}=\Xi^{+}\,\mathcal{K}_{\gamma_{3}}^{2}
\mathcal{K}_{2\gamma_{3}}^{-2}\,\Xi^{-}\,\mathcal{K}_{\gamma_{1}}\]
\end{proof}

\begin{remark}
The proposition 4.11 above yields the desired "jump" of  $\Omega(\gamma_{3})=\Omega((\alpha_{3},\alpha_{1}+(n-1)\alpha_{3}))$, which is compatible with the \textit{"primitive wall crossing formula"}(see for example \cite{manschot2011wall}) that gives the "jump"
\begin{equation}
    \Delta\Omega(\gamma_{3}\to\gamma_{1}+\gamma_{2})=(-1)^{\langle\gamma_{1},\gamma_{2}\rangle+1}|\langle\gamma_{1},\gamma_{2}\rangle|\,\Omega(\gamma_{1})\,\Omega(\gamma_{2})=1
\end{equation}

\end{remark}

\begin{proposition}
In parallel with the above argument, since we have that \[(-\alpha_{3},\alpha_{2}+(n-1)\alpha_{3})=(-\alpha_{1},(n-1)\alpha_{1})+(-\alpha_{2},n\alpha_{2})\]
we infer that on one side of $\mathcal{W}_{\alpha_{1},\alpha_{2}}^{1}$, where $\mathcal{W}_{k}^{1}$'s are present, the states with charge  $(-\alpha_{3},\alpha_{2}+(n-1)\alpha_{3})$ exist with DT-invariant equals to $1$, while on the other side, the same states derive their existence in two ways. Thus, we expect the DT-invariant $\Omega((-\alpha_{3},\alpha_{2}+(n-1)\alpha_{3}))$  to "jump" from $1$ to $2$ after crossing the wall $\mathcal{W}_{\alpha_{1},\alpha_{2}}^{1}$. 
\end{proposition}

\begin{proof}
The same as the proof of the proposition 4.11 above.
\end{proof}

Yet, there exist states that are present only on one side of $\mathcal{W}_{\alpha_{1},\alpha_{2}}^{1}$, which are "scattered" by $(\alpha_{1},n\alpha_{1})$ and $(\alpha_{2},(n+1)\alpha_{2})$. Indeed, by applying pentagon identity, we see that after crossing the wall, new states with charge $(\alpha_{3},\alpha_{2}+n\alpha_{3})$ exist with DT invariants one. The BPS spectrum for pure $SU(3)$ gauge theory thus obtained by using WCS is consistent with that obtained by physical approaches (\cite{fraser1997weak}\cite{hollowood1997strong}\cite{kuchiev2008charges}).

\section{Acknowledgement}

I would like to thank my academic advisor Professor Yan Soibelman for giving me this research topic and for his constant support during my Ph.D study. Without his valuable advises on how to conduct research and his stimulating instructions on various topics, this paper can never come into existence.

\bibliographystyle{unsrt}
\bibliography{references}
\end{document}